\documentclass[11pt]{article}
\usepackage[margin=1in]{geometry}
\usepackage[T1]{fontenc}
\usepackage{lmodern,microtype,times}
\usepackage{amsmath,amssymb,amsthm,mathtools}
\usepackage{enumitem}
\usepackage{xcolor}
\usepackage[hidelinks]{hyperref}

\hypersetup{pdftitle={A sampling Lov\'{a}sz Local Lemma}}
\newtheorem{theorem}{Theorem}
\newtheorem{lemma}{Lemma}
\newtheorem{proposition}{Proposition}

\newcommand{\one}{\mathbf 1} 
\newcommand{\vbl}{\operatorname{vbl}}
\newcommand{\Pairs}{\mathcal P}
\newcommand{\TV}{d_{\mathrm{TV}}}
\title{A sampling Lov\'{a}sz Local Lemma}
\author{Dimitris Achlioptas\\
University of Athens\\
{\small{\texttt{optas@uoa.gr}}}} 
\date{}

\begin{document}
\maketitle
\begin{abstract}
We give an approximately uniform sampler for satisfying assignments of constraint satisfaction problems that satisfy $4\mathrm e p(\Delta+1)^2\le1$, where $p$ is the largest constraint-violation probability under the uniform product distribution, and $\Delta$ is the maximum degree of the dependency graph. The algorithm invokes the recent efficient approximate counting algorithm of Liu, Wang, Yin, Zhang, and Zhou~\cite{CountingLLL} as a subroutine and returns a satisfying assignment sampled within total-variation distance $\varepsilon$ of the uniform distribution in $(n+m/\varepsilon)^{O(k\Delta\log D)}$ time, where $n$ and $m$ are the numbers of variables and constraints, $D$ is the common domain size, and $k$ bounds the constraint arity. 
We also give an asymmetric product-form condition under which both
counting and sampling are efficient. To supply the counts needed in the
additional regime, we extend the marginal-expansion proof of Liu et al.
\end{abstract}

\section{Introduction}

Given an arbitrary constraint satisfaction problem (CSP), let $p$ be an upper bound on the fraction of violating value assignments of any constraint and assume that each constraint shares variables with at most $\Delta$ others.  Very recently, Liu, Wang, Yin, Zhang, and Zhou~\cite{CountingLLL}  gave efficient algorithms for approximately counting solutions whenever $4\mathrm e p(\Delta+1)^2\le1$, a condition we will call the Lov\'asz local lemma counting regime. They left open whether one can sample in the same
regime~\cite[Section~6]{CountingLLL}. We answer this question affirmatively. 

The difficulty in going from counting to sampling is that repeatedly setting variables according to their conditional marginals may well destroy the hypothesis enabling counting.
For example, a $k$-clause in a SAT instance has violation probability $2^{-k}$ initially, but after unfavorably fixing $k-1$ of its variables  the probability becomes $1/2$, while there is no guarantee that its interactions with other clauses have been reduced correspondingly.  We overcome this problem by invoking the counting algorithm of~\cite{CountingLLL}  only on instances obtained by deleting constraints of the original input instance, inheriting its guarantee $4\mathrm e p(\Delta+1)^2\le1$.

\subsection{Results}

A CSP $H$ is a set of $m\ge1$ constraints on $n$ variables, each taking values in $[D]=\{1,\ldots,D\}$. We assume $D\ge2$ and that each constraint is a nonconstant Boolean predicate on a nonempty scope of at most $k$ variables. The dependency graph $G$ has a vertex for each constraint and two constraints are adjacent  if and only if they share variables, i.e., their scopes intersect. For $j=1,2$, let $N_j(c)$ be the constraints at distance exactly $j$ from $c$ in $G$, and set $\Delta=\max_c|N_1(c)|$ and $\Delta_2=\max_c|N_2(c)|$. Throughout we assume $\Delta,\Delta_2\ge2$. Let $p$ be the maximum constraint-violation
probability under independent uniform variable values. For $K\subseteq H$, let $Z(K)$ be the number of
assignments to all $n$ variables satisfying $K$, and let $\mu_K$ be
the uniform distribution on these assignments when $Z(K)>0$.
We write $\mathcal L(X)$ for the law of a random variable $X$ and
$\TV(\mu,\nu)=\frac12\sum_x|\mu(x)-\nu(x)|$ for the total-variation
distance between probability measures $\mu,\nu$. 

\begin{theorem}[Sampling in the counting LLL regime]\label{thm:main}
There exists a randomized algorithm which given any CSP instance $H$ satisfying $4\mathrm e p(\Delta+1)^2\le1$ and any $0<\varepsilon\le1/2$,  returns a satisfying assignment $X$ of $H$ with $\TV(\mathcal L(X),\mu_H)\le\varepsilon$ after at most $(n+m/\varepsilon)^{O(k\Delta\log D)}$ steps.
\end{theorem}
For example, for Boolean $k$-CNF formulae, the condition becomes
\[
 \Delta+1\le \frac{2^{k/2}}{2\sqrt{\mathrm e}} \enspace .
\]

Counting  with relative accuracy $\eta$ returns a value $\widehat Z(K)$ satisfying $(1-\eta)Z(K)\le\widehat Z(K)\le(1+\eta)Z(K)$. The condition in Theorem~\ref{thm:main} is the sufficient condition of~\cite{CountingLLL} for efficient such approximate counting. Our sampler-given-a-counter result actually requires the weaker condition $8pQ\le1$, where $Q$ measures the actual first and second neighborhoods:
\begin{equation}\label{eq:parameters}
 Q=1+\max_c\bigl(|N_1(c)|+|N_2(c)|\bigr)
       \le1+\Delta+\Delta_2.
\end{equation}

\begin{theorem}[Deletion-only reduction]\label{thm:oracle}
There exists a randomized algorithm which given any CSP instance $H$ satisfying $8pQ\le1$, any $0<\varepsilon\le1/2$, and access to a  relative-approximation counter for $Z(K)$ for any $K\subseteq H$,  returns a satisfying assignment $X$ of $H$  with $\TV(\mathcal L(X),\mu_H)\le\varepsilon$ after at most $(m/\varepsilon)^{O(\log\Delta_2)}$ counter invocations of relative accuracy $\varepsilon/(32m)$ and $(n+m/\varepsilon)^{O(k\Delta\log D)}$ additional steps.
\end{theorem}

Theorem~\ref{thm:main} follows from Theorem~\ref{thm:oracle} since $\Delta_2\le\Delta(\Delta-1)$, and hence
\begin{equation}\label{eq:closure}
 Q\le1+\Delta^2\le(\Delta+1)^2,
 \qquad 4\mathrm e p(\Delta+1)^2\le1
       \ \Longrightarrow\ 8pQ\le2/\mathrm e<1.
\end{equation}
The deterministic counter of~\cite[Theorem~1.2]{CountingLLL} runs in
$(n\Delta/\eta)^{O(k\Delta\log D)}$ time at relative accuracy
$\eta$. Applying
it to the queries of Theorem~\ref{thm:oracle} with
$\eta=\varepsilon/(32m)$ gives the claimed total running time.

\subsection{An asymmetric extension}

Let $p_c$ be the individual violation probability of constraint $c$, and
write
\[
 \mathcal B(c)=\{c\}\cup N_1(c)\cup N_2(c).
\]
The same insertion flow also gives a reduction under an asymmetric
product condition.

\begin{theorem}[Asymmetric deletion-only reduction]\label{thm:asymmetric}
Fix $0<C<1/4$. If there are $x_c\in(0,1)$ such that
\[
 p_c\le Cx_c\prod_{b\in\mathcal B(c)}(1-x_b)
 \qquad(c\in H) ,
\]
then the conclusion of Theorem~\ref{thm:oracle} holds, with the constants in its resource bounds allowed to depend on $C$.
\end{theorem}

The counting theorem of Liu et al.~\cite{CountingLLL} supplies the
required deletion counts in its stated symmetric regime. To supply them
in the additional regime of Theorem~\ref{thm:asymmetric}, we extend their
counting proof. Their marginal expansion and size-budget truncation are
our foundation; the additional estimates control the same recursion
under the product condition.

\begin{theorem}[Counting under the asymmetric condition]\label{thm:asymmetric-counting}
Under the hypothesis of Theorem~\ref{thm:asymmetric}, a deterministic
algorithm returns, for every $K\subseteq H$ and $0<\delta\le1/2$, a
positive rational $\widehat Z(K)$ satisfying
\[
 (1-\delta)Z(K)\le\widehat Z(K)\le(1+\delta)Z(K)
\]
in $(n+m/\delta)^{O(k\Delta\log D)}$ time. The constant in the exponent
is absolute, uniformly over $0<C<1/4$.
\end{theorem}

Thus the asymmetric sampler can be implemented in total time
$(n+m/\varepsilon)^{O_C(k\Delta\log D)}$, including its counting calls.
For the symmetric specialization, however, we use the new counter with
our original symmetric reduction. This gives the following complete
algorithm without a dependence on the gap below $C=1/4$.

\begin{theorem}[Symmetric counting and sampling]\label{thm:symmetric-complete}
Suppose
\[
 p<\frac{Q^Q}{4(Q+1)^{Q+1}}.
\]
For every $0<\delta,\varepsilon\le1/2$, there is a deterministic
relative-$\delta$ approximation to $Z(H)$ in
$(n+m/\delta)^{O(k\Delta\log D)}$ time, and a randomized algorithm
that always returns a satisfying assignment $X$ with
$\TV(\mathcal L(X),\mu_H)\le\varepsilon$ in
$(n+m/\varepsilon)^{O(k\Delta\log D)}$ total time.
\end{theorem}

In particular, $4\mathrm e p(Q+1)\le1$ suffices.
Theorem~\ref{thm:symmetric-complete}, compared with combining
\cite{CountingLLL} and Theorem~\ref{thm:oracle},  replaces
$(\Delta+1)^2$ by $Q+1$ and gives a small finite-$Q$ refinement. In the
worst-case, $Q\sim\Delta^2$ and  the leading constant remains
$4\mathrm e$ so that the full symmetric counting-and-sampling guarantee is
 essentially the one already obtained without the asymmetric
extension. The replacement of $(\Delta+1)^2$ by $Q+1$ is motivated by the fact that the latter can be significantly smaller  in graphs with many small cycles such as lattices.

\section{Related work}\label{sec:related}

Two difficulties recur in sampling under local-lemma conditions:
local repair need not preserve the uniform distribution, and fixing
variables can destroy the hypotheses used to analyze the remaining
instance. The closest  work to ours uses constraint-wise comparisons, organizing the computation around inserting one constraint at a time. We begin there, and then place that approach among the main sampling methods. Recall that a constraint is
atomic if it forbids exactly one assignment to its variables.

\subsection{Constraint-wise methods for sampling and the Counting LLL}

The closest preceding approaches compare formulas differing by a single
constraint. Specifically, Wang and
Yin~\cite{WangYin} developed such a coupling for atomic CSPs and used
it for deterministic counting and randomized sampling. For common
domain size $D$, their sufficient condition is
\[
 (8\mathrm e)^3p(\Delta+1)^{2+\phi(D)}\le1,
 \qquad
 \phi(D)=\frac{2\log(2-1/D)}{\log D-\log(2-1/D)}.
\]
The exponent tends to $2$ as $D$ grows and is about $4.82$ for Boolean variable domains.  

Liu, Wang, Yin, Zhang, and Zhou~\cite{CountingLLL} reached $4\mathrm e p(\Delta+1)^2\le1$ for arbitrary predicates using an expansion over 2-trees, i.e., sets of constraints that are independent in the dependency graph $G$ and connected in $G^2$. These structures, which also underlie our work, have a long history in the algorithmic LLL literature going back to the work of Alon~\cite{Alon}, while these and related structures appear in marking, coupling, and recursive algorithms~\cite{Moitra,FGYZ,JPVGeneral,HWYCounting,WangYin,CountingLLL}. The deterministic counter of~\cite{CountingLLL} takes
$(n\Delta/\delta)^{O(k\Delta\log D)}$ time at relative accuracy $\delta$, while their randomized algorithm takes $\operatorname{poly}(k,\Delta,D)(n/\delta)^2$ time and succeeds with probability at least $3/4$. The work~\cite{CountingLLL} does not provide a sampling algorithm for this regime and leaves the question open.

\subsection{Marginals, marking, and general predicates}

Moitra's groundbreaking work~\cite{Moitra} was the first to give approximate counting and sampling algorithms under LLL-like conditions. His
method estimates variable marginals by coupling two conditioned
distributions. A marking of the variables ensures that revealing
values leaves enough randomness in each constraint; the region of
disagreement can then be truncated and represented by a linear
program. Marking also addresses the failure of ordinary
self-reduction after arbitrary pinnings.

Guo, Liao, Lu, and Zhang~\cite{GLLZ} extended this approach to
hypergraph colorings and non-Boolean domains. Feng, Guo, Yin, and
Zhang~\cite{FGYZ} improved the running time for bounded-occurrence
$k$-SAT by sampling a marked projection and then a compatible
completion. Their mixing argument applies to the projection even
when single-variable moves need not connect the full solution space.
State compression, introduced by Feng, He, and Yin~\cite{FHY},
generalizes marking by mapping each domain to a smaller one. This
allows projected sampling for atomic CSPs and for constraints with
bounded numbers of forbidden assignments.

For arbitrary predicates, Jain, Pham, and Vuong~\cite{JPVGeneral}
proved deterministic approximate counting and approximate sampling
under $ kD^3p\Delta^7<c_0,$
where $c_0>0$ is absolute, with polynomial running time for fixed
structural parameters. Their separate work on atomic
CSPs~\cite{JPVAtomic} uses information percolation to analyze projected
dynamics and improve the parameter range and sampling time.

He, Wang, and Yin~\cite{HWYSampling} gave an expected near-linear-time
approximate sampler for general predicates under $  kD^2p\Delta^5\le c_1$
for a sufficiently small absolute $c_1>0$. Rather than running a
Markov chain, their marginal sampler recursively generates the
random dependencies needed to answer a query. The near-linear bound
uses an oracle deciding whether a partial local assignment already
forces a constraint to be satisfied. Their deterministic counting
algorithm~\cite{HWYCounting} derandomizes this construction in the
same parameter regime. It avoids linear programming and uses
generalized $\{2,3\}$-tree witnesses to bound the error, with a
polynomial running time for fixed structural parameters.

\section{How our algorithm works}\label{sec:overview}

In this section we give an informal, intuitive description of our algorithm and its analysis which we then flesh out formally in the remaining sections.

We start by describing a perfect sampling algorithm for any satisfiable CSP, assuming exact deletion counts, i.e., counts for subinstances obtained by removing constraints (without fixing any variable values). The sampler takes $m$ rounds, one per constraint, but each round may contain an unbounded number of steps and each step may need to consider an exponentially large set of options. We will then make this sampler efficient by restricting the available options to polynomially many, imposing a polynomial cap on the total number of steps across all rounds of a trial, and replacing exact counts by approximate ones. Approximate uniformity will follow from $8pQ \le 1$. 

The perfect sampler begins with the empty set of constraints and a uniformly random assignment in $[D]^n$.  It adds constraints one by one in an arbitrary fixed order. Let $F_i$ comprise the first $i$ constraints. After round $i$, the current assignment $\sigma_i$ will be uniformly distributed over the solutions of $F_i$. To make progress, the algorithm adds the next constraint and starts a random walk from $\sigma_i$, eventually stopping at a uniformly random solution $\sigma_{i+1}$ of $F_{i+1}$.

Three things are worth noting right away. First, the starting assignment is a uniformly random solution of $F_i$, not merely a solution; this warm start is crucial to the analysis. Second, the walk takes place on the full assignment space $[D]^n$ and may, thus, pass through states that violate $F_i$. Third, reaching a solution of $F_{i+1}$ does not necessarily end the walk. Stopping at the first such solution can introduce bias, so the walk may visit the solution set of $F_{i+1}$ several times before stopping.

The general idea  is that, under the LLL condition, inserting a constraint can usually be accommodated by only changing the variables of constraints that are nearby in the dependency graph. Of course, maintaining uniformity is more demanding than mere accommodation, i.e., than finding a nearby solution. To see what must be accomplished, fix one insertion and write $F_{i+1}=F_i\cup\{c\}$. The surviving solutions already have equal probability under the uniform distribution on solutions of $F_i$. Thus, the task is to transport and redistribute equally among the survivors the mass assigned to the solutions eliminated by $c$. 

We think of this redistribution as a flow problem on $[D]^n$. The solutions of $F_i$ that violate $c$ supply net outflow, the assignments that violate $F_i$ are transportation vertices, and the solutions of $F_{i+1}$ receive net inflow. These last vertices need not be pure sinks, as the walk may continue from them. Given such a flow, a flow-balance argument for  absorbing Markov chains will give a stopping rule whose final assignment is uniform on the solutions of $F_{i+1}$. We will construct the flow from families of local variable mutations, chosen with probabilities that can be computed using only deletion counts. There are three tasks to achieve this.

\subsection{Grouping the transitions}

The first task is to give each value assignment (state)  a bag of weighted outgoing transitions, grouped into families. Each family is indexed by a set of constraints, which we call the skeleton of the family. The skeleton determines the mutable variables and which changes (mutations) of the mutable variables are allowed. It is important to note that while the skeleton uniquely determines the mutable variables, the set may include variables outside the skeleton's own constraints. For a given skeleton, let $\Lambda$ be its set of mutable variables and let $M$ be the subinstance obtained from $F_{i+1}$ by deleting every constraint with a variable in $\Lambda$. Each of the $D^{|\Lambda|}$ mutations is either forbidden or has weight equal to the number of assignments to the variables outside $\Lambda$ that satisfy $M$, a deletion-count assumed to be available. 

At each step, the algorithm chooses a family with probability proportional to the total weight of its permitted mutations, and then chooses uniformly among those mutations (since, as we saw, they all have the same weight). If the current assignment satisfies $F_{i+1}$, a stopping option of weight $Z(F_i)$ competes with the family choices. A family with no permitted mutation has weight zero. Each such choice, including stopping, counts as one walk step. 

\subsection{Growing the skeleton and specifying its mutations}

The second task is to describe these families. All neighborhoods in this description are taken in the dependency graph of $F_{i+1}$.  We write $A$ for the skeleton and $U$ for the constraints whose variables will comprise the mutable region. Initially, $A=\{c\}$, while $U$ contains $c$ and all its immediate neighbors. Note that these neighbors are included in the mutable region, but not in the skeleton. We will call $c$ the root of the skeleton. We see that  constraints at distance two from the root $c$ can be affected by changes to variables appearing in $U$.

A procedure next grows the mutable region by considering constraints outside $U$ whose scopes intersect the variables appearing in $U$ and whose role has not yet been decided. Specifically, for the lowest such constraint $e$, the procedure explores two branches. In one branch, expansion stops at $e$ and $e$ itself is placed in a boundary set $C$. In the other branch, the procedure expands through $e$ and $e$  is added to the skeleton and also to $U$ along with its neighbors (except for those neighbors already assigned to $C$). Newly added constraints to $U$ may now intersect further unclassified constraints, so  the procedure continues until every constraint sharing a variable with a constraint in $U$ belongs to either  $U$ or $C$. Thus,  every generated skeleton is an independent set in the dependency graph $G$ and connected in $G^2$, the graph joining constraints at distance at most two in $G$.  The  executions of the procedure are in one-to-one correspondence with the generated skeletons: given the global constraint order and a  skeleton $A$, replaying the procedure recovers both $U$ and $C$. 

For a skeleton $A$,  let  $B=U\setminus A$. There are thus four kinds of constraints: those in the skeleton
$A$, the other constraints $B$ of the mutable region $U$, the boundary constraints $C$, and 
the  remainder constraints $M$, disjoint from $U$. The mutable variables are 
$\Lambda=\vbl(U)=\vbl(A\cup B)$, where $\vbl$ denotes the union of
the indicated scopes.  

For a mutation $x\to y$ of the variables in $\Lambda$ to be allowed, $x$ and $y$  must agree outside $\Lambda$ and must both satisfy every constraint in $B\cup M$.  Note that in every skeleton, the neighbors of $c$ belong to $B$  by the initialization rule. Since the walk starts at a solution of $F_i$, it follows that the neighbors of $c$  remain satisfied throughout the walk's trajectory. Regarding the skeleton's own constraints, the requirement is that $c$ must change its satisfaction status, while every other skeleton constraint must be violated by at least one of $x,y$, i.e., it may be repaired, become violated, or remain violated, but not remain satisfied. We call this last requirement the \emph{two-endpoint violation test}. No satisfaction status requirement is  placed on the boundary constraints $C$, which are ``abandoned to their fate''. This abandonment allows these constraints to act as insulation between the mutable region $U$ and the rest of the instance $M$, a non-interaction which is at the heart of being able to use deletion counts to weigh transitions. The price for this insulation is that while all the constraints in $C$ are satisfied when the walk starts, no effort is made to keep them so. 

The inserted constraint $c$  imposes a special requirement that depends on the parity of the skeleton's size $|A|$, i.e., on the number of constraints in $A$. An odd-sized skeleton can be used only when $c$ is violated, and only mutations that repair $c$ are allowed. An even-sized skeleton can be used only when $c$ is satisfied, and only mutations that violate $c$ are allowed. These reverse moves are correction terms from an
inclusion--exclusion identity on pairs of assignments. Reversing the
arcs reverses their contribution to net flow, allowing alternating
signs to become directions, while all transition weights remain
nonnegative. The two-endpoint tests, together with these orientations and the deletion-count weights, ensure that the local flows have exactly the required net balance when added together. Individual mutations need not reduce the number of violated constraints. It is their combined balance that gives the correct output distribution.

\subsection{Making the process efficient}

As mentioned, every  skeleton is an independent set of $G$ connected in $G^2$. Thus, a skeleton may contain a constant fraction of all constraints, and there may be exponentially many skeletons. Clearly, we cannot afford to enumerate them before each step in order to select one. Thus, the third task is to deal with large skeletons. 

The mutable region of a skeleton of size $r$ has at most $k(\Delta+1)r$ variables. A skeleton size cutoff then bounds the assignments to enumerate within each mutable region. The key quantity that will allow us to disregard a skeleton concerns how much traffic its family carries over the entire walk for inserting $c$. Specifically, sum the weights of all the arcs in the family and divide by $Z(F_i)Z(F_{i+1})$. The flow argument bounds the expected total number of times the walk uses this family by that normalized sum, including repeated uses. In particular, this also bounds the probability that the family is ever used.

Under $8pQ\le1$, we will show that the normalized total weight of a family decays exponentially in its skeleton size. The calculation is given in Section~\ref{sec:singlepair}.

This decay competes with the exponential growth in the number of skeletons as a function of size. Connectivity in $G^2$ bounds that growth, analogously to connectivity in $G$ bounding the growth of witness-trees in the Moser--Tardos analysis~\cite{MT}. In addition, all deleted constraints---the skeleton, its included neighbors, and the insulating boundary---lie within distance two of the skeleton. These are the two places where the second neighborhood enters the analysis. Our LLL condition ensures that the decay wins, so the sum of normalized family weights above a given skeleton size has a geometric tail. The expected-use bound then says that the whole execution is unlikely ever to choose a large skeleton. We do not need a comparable bound separately at every possible intermediate state; the uniform starting distribution, i.e., the warm start, is what makes the aggregate estimate useful.

Consequently, the algorithm can retain only skeletons of size logarithmic in $m/\varepsilon$. The same argument bounds the expected number of steps, which lets us bound the failure probability due to imposing a limit on the total steps across all insertions of a trial. Finally, the algorithm uses sufficiently accurate relative approximations to the deletion counts. A coupling with the ideal algorithm controls the error from omitting large mutations, limiting the running time, and approximating the remaining weights. 

\section{Constraint insertion by a flow}\label{sec:absorption}

Let $V$ be the variable set and put $\Omega=[D]^V$. Recall that $Z(K)$ counts assignments
to all of $V$, so $Z(\varnothing)=D^n$. Fix the insertion order
$c_1,\ldots,c_m$, with $F_0=\varnothing$ and
$F_i=\{c_1,\ldots,c_i\}$.

To analyze one constraint insertion, fix $F\subseteq H$, $c\in H\setminus F$, write
$J=F\cup\{c\}$, and assume $Z(J)>0$.  Write $\one_E$ for the indicator of
an event $E$. For finite nonnegative arc weights $W(x,y)$ on $\Omega$, put
$o(x)=\sum_yW(x,y)$ and $i(x)=\sum_yW(y,x)$.
We ask for the following net outflow at every assignment $x$:
\begin{equation}\label{eq:balance}
 o(x)-i(x)=Z(J)\one_{\{x\models F\}}-Z(F)\one_{\{x\models J\}}.
\end{equation}
At a solution of $F$ that violates $c$, the required net outflow is
$Z(J)$. At a solution of $J$, the required net inflow is
$Z(F)-Z(J)$. At every other assignment, inflow and outflow must agree.
After dividing by $Z(F)Z(J)$, these amounts are exactly the excess
and deficit between $\mu_F$ and $\mu_J$.

Give each solution of $J$ stopping weight $\tau(x)=Z(F)$ and set $\tau(x)=0$ at every other assignment. At each visit, the algorithm chooses between stopping and following each outgoing arc in proportion to its weight. The lemma below shows both that this rule gives the desired output law and that the expected number of visits to each assignment is bounded. After its proof, we use the visit bound to control how often the walk selects each mutable region.

\begin{lemma}[Absorption and occupation]\label{lem:absorption}
Let $a,\tau:\Omega\to[0,\infty)$ have the same sum $M_*>0$ over $\Omega$ and suppose $o-i=a-\tau$. Set $d(x)=o(x)+\tau(x)$. Define a walk on $\Omega$ whose initial state has law $a/M_*$ and which, at each $x\in\Omega$, stops with probability $\tau(x)/d(x)$ or moves to $y\in\Omega$ with probability $W(x,y)/d(x)$.

The  walk is well defined, i.e., $d(x)>0$ for all reachable $x$, stops almost surely, and has output law $\tau/M_*$. The expected number of  visits to each $x\in\Omega$ is at most $d(x)/M_*$.
\end{lemma}
\begin{proof}
We first bound the expected visits, and then use that bound to identify
the stopping distribution. Balance gives
\begin{equation}\label{eq:capacity}
 d(x)=a(x)+\sum_yW(y,x).
\end{equation}
Thus $d(x)=0$ only at states with neither initial mass nor an incoming
positive-weight arc, which are unreachable. Put $v_0=a$, and let
$v_j(x)/M_*$ be the probability  that the walk has not yet stopped and is at  $x$ immediately before step $j$, with steps numbered from $0$. Then
\[
 v_{j+1}(y)=\sum_{x:d(x)>0}v_j(x)\frac{W(x,y)}{d(x)}.
\]
Summing $v_j(x)/M_*$ over $j$ counts all visits to $x$,
not just the probability of reaching it once. We claim that
\begin{equation}\label{eq:visits}
 \sum_{j=0}^r v_j(x)\le d(x)\qquad(r\ge0).
\end{equation}
The case $r=0$ follows from~\eqref{eq:capacity}, and induction gives
\[
 \sum_{j=0}^{r+1}v_j(y)
 =a(y)+\sum_{x:d(x)>0}\left(\sum_{j=0}^rv_j(x)\right)
                \frac{W(x,y)}{d(x)}
 \le a(y)+\sum_xW(x,y)=d(y).
\]
After division by $M_*$, this is the expected-visit bound. Summing it
over $x$ shows that the expected total number of steps is finite,
so the process stops almost surely. At a state with $d(x)>0$, the total unnormalized mass stopped at $x$ is $\tau(x)\sum_jv_j(x)/d(x)\le\tau(x)$. States with $d(x)=0$ have $\tau(x)=0$ and are never visited. All $M_*$ units of initial mass eventually stop, while $\sum_x\tau(x)=M_*$. Thus every upper bound is attained, and division by $M_*$ gives the claimed output law.
\end{proof}

For the constraint insertion $F\to J$, apply Lemma~\ref{lem:absorption} with $a(x)=Z(J)\one_{\{x\models F\}}$ and
$M_*=Z(F)Z(J)$. The initial and terminal distributions in the lemma
are then $\mu_F$ and $\mu_J$. The balance equation and stopping
rule give the required distribution.

We construct the flow $W$ as a sum of nonnegative mini-flows, one for each family of mutations (skeleton). Write $P$ for a family label and $W_P$ for its arc weights, so $W=\sum_PW_P$. Put $o_P(x)=\sum_yW_P(x,y)$.  At $x$, the algorithm chooses family $P$ with probability $o_P(x)/d(x)$ and, conditional on that choice, chooses its destination $y$ proportionally to $W_P(x,y)$. Stopping still has probability $\tau(x)/d(x)$. This is the same walk on assignments, with an additional record of the family supplying each move. Different families can supply the same arc; their weights add.

Let $V_x$ count all visits to $x$ during this insertion, including returns and the visit at which the walk stops. We want to count every selection of $P$, not just the event that it is selected once. In the notation of the lemma's proof, $v_j(x)/M_*$ is the probability that step $j$ is reached at $x$. At that step the conditional probability of selecting $P$ is $o_P(x)/d(x)$. Thus
\[
 \begin{aligned}
 \mathbb E[\text{selections of }P]
 &=\sum_{j\ge0}\sum_{x:d(x)>0}
       \frac{v_j(x)}{M_*}\frac{o_P(x)}{d(x)}\\
 &=\sum_{x:d(x)>0}\mathbb E[V_x]\frac{o_P(x)}{d(x)}.
 \end{aligned}
\]
This is a sum of conditional choice probabilities, not an independence assumption: an earlier choice of $P$ may affect whether the walk returns to $x$. All summands are nonnegative, so the sums can be interchanged. Applying $\mathbb E[V_x]\le d(x)/M_*$ and writing $\|W_P\|_1=\sum_{x,y}W_P(x,y)$ gives
\begin{equation}\label{eq:usage}
 \mathbb E[\text{selections of }P]
 \le \sum_x\frac{d(x)}{M_*}\frac{o_P(x)}{d(x)}
 =\frac{\|W_P\|_1}{M_*}.
\end{equation}
Terms with $d(x)=0$ are taken to be zero; then $o_P(x)=0$ as well. The factor $d(x)$ in the visit bound cancels the denominator of the choice probability. There is no further factor for the length of the walk: the sum over $j$ has already counted every opportunity to choose $P$. Unreachable states with $d(x)>0$ may contribute to the upper bound. If $K_{\mathcal A}$ counts selections from a collection $\mathcal A$ of labels, repeated selections count separately. The event of at least one selection has indicator at most $K_{\mathcal A}$. Taking expectations and summing~\eqref{eq:usage} gives
\[
 \Pr[K_{\mathcal A}\ge1]\le\mathbb E K_{\mathcal A}
 \le\frac1{M_*}\sum_{P\in\mathcal A}\|W_P\|_1.
\]
Thus a family with small normalized total weight is unlikely ever
to be used, even if the walk revisits some assignments many times.
We will apply this to all families with large skeletons. No bound on
their selection probability at every possible state is needed, but
the starting law $a/M_*$ is essential and a sample from it is supplied by the uniform solution of $F_i$. In the ideal sampler, each round ends with an exactly uniform solution, supplying the next insertion with this law. Section~\ref{sec:trial-proof} compares the entire implemented trial with this ideal execution: the probability of their first discrepancy is bounded along the ideal trajectory. This accounts for errors accumulated in earlier insertions as well as errors introduced by the current one.

\section{Constructing the insertion flow}\label{sec:flow}

We now construct the families described in Section~\ref{sec:overview}. Their weights count outside completions, and their directions turn the signs of a paired inclusion--exclusion identity into nonnegative local flows.

\subsection{Choosing the variables of a mutation}

Recall that for a set of constraints $U$ we write $\vbl(U)=\bigcup_{d\in U}\vbl(d)$. All neighborhoods below are in the induced graph $G[J]$. For $U\subseteq J$, let $N(U)=N_{G[J]}(U)$ be the constraints outside $U$ that are adjacent to $U$, and write $N(e)=N(\{e\})$. As in the overview, $U$ is the mutable region and $C$ is the boundary where expansion has stopped. We call a constraint in $N(U)\setminus C$ \emph{vulnerable}.

Fix an order of the constraints. The following deterministic branching rule generates the skeletons and their associated mutable regions. Its output depends only on the formula, not on the current assignment:
\par
\begin{samepage}
\begin{enumerate}[leftmargin=*,label=\arabic*.]
\item Initialize $A = \{c\}$, $U=\{c\}\cup N(c)$ and $C=\varnothing$.
\item If a vulnerable constraint exists, take the least one $e$ in the
fixed order. In one branch, add $e$ to $C$ and leave $A,U$ unchanged. In the other, leave $C$ unchanged and expand the mutable region through $e$ by setting
\begin{eqnarray*}
 A & \leftarrow & A \cup \{e\} \\
 U & \leftarrow & U\cup\{e\}\cup(N(e)\setminus C)
\end{eqnarray*}
Continue in both branches.
\item When no vulnerable constraint remains, let $B=U\setminus A$ and output $(A,B)$.
\end{enumerate}
\end{samepage}
The pair $(A,B)$ records the completed expansion of the mutable region from $c$. Let $\Pairs(J,c)$ denote the set of pairs $(A,B)$ produced by the rule; by the following lemma, each is determined by its skeleton $A$.

\begin{lemma}[Structure]\label{lem:structure}
The constraints in $A$ have disjoint scopes, while the subgraph of $G$ induced by $A \cup B$ is connected.  Each expansion has a two-edge connection to $c$ or to an earlier expansion constraint, with intermediate constraint in $B$. The set $A$ determines $B$ and at termination $C=N(A\cup B)$.
\end{lemma}
\begin{proof}
Initially and by subsequent construction,  $A\subseteq U$ and $U\cap C=\varnothing$. The initial $U$ is connected, and each expansion adds some neighbor of $U$ and some of that neighbor's neighbors. Every neighbor of  a constraint already in $A$ lies in $U\cup C$. Hence a constraint $e$ selected for expansion, being outside $U\cup C$, is adjacent to none of them. This proves disjointness of the scopes in $A$. A neighbor of $e$ in the previous $U$ is not in $A$; it entered $U$ as a neighbor of an earlier member of $A$. Once in $U$ it cannot later be selected for expansion, so it belongs to $B$ and  supplies the stated two-edge path. Given a generated final $A$, replay the rule, expanding iff the selected constraint belongs to $A$. This recovers the branch and, thus, $B$. Every constraint put in $C$ remains an external neighbor of $U$, and
termination leaves no such neighbor outside $C$. Hence $C=N(U)$.
\end{proof}

Recall that $G^2$ joins vertices at distance at most two in $G$.
For $A\subseteq H$, write $G^2[A]$ for its induced graph on $A$. A skeleton is a
nonempty independent set $A$ of $G$ with $G^2[A]$ connected. By Lemma~\ref{lem:structure}, every generated $A$ is a skeleton of
$G$, even when the construction is run on a subformula. These are the
$2$-trees used in~\cite{Alon,WangYin,CountingLLL}. The structure has three uses. Disjoint skeleton scopes give the small probability factors in the tail bound; distance-two connectivity bounds the number of skeletons; and the partition of subsets by completed branches, proved below, gives the exact inclusion--exclusion grouping. The present rule
stops expansion at the boundary of the variables being changed.
Its boundary is different from that in the canonical construction
of~\cite{CountingLLL}, and we establish the corresponding grouping
identity below.
The skeleton measures the size of an expansion. The actual mutation
also includes the surrounding constraints in $B$. For a generated
pair define
\begin{equation}\label{eq:pairdata}
 \Lambda=\vbl(A\cup B),\qquad
 C=N_{G[J]}(A\cup B),\qquad M=J\setminus(A\cup B\cup C).
\end{equation}
Writing $F=J\setminus\{c\}$, we have
\begin{equation}\label{eq:separation}
 F=(A\setminus\{c\})\sqcup B\sqcup C\sqcup M,
 \qquad \vbl(M)\cap\Lambda=\varnothing.
\end{equation}

\subsection{The weights and directions of the mini-flows}

Fix a  skeleton $A$ and its associated pair $(A,B)$. A permitted mutation will have the same weight for every compatible choice of the immutable variables.
The relevant multiplicity is the number of outside assignments
satisfying $M$. Set
\begin{equation}\label{eq:h}
 h(A,B)=D^{-|\Lambda|}Z(M).
\end{equation}
Since $M$ does not mention the variables in $\Lambda$, its full
count $Z(M)$ includes $D^{|\Lambda|}$ arbitrary choices for those
variables. Dividing them out leaves exactly the number of satisfying
outside assignments. This number is positive because $J$ is satisfiable.

For a constraint $d$, the two-endpoint violation test on assignments
$x,u$ asks whether at least one of them violates $d$. Excluding this
event for every $d\in F$ is exactly the requirement that both
assignments satisfy $F$. This is the event to which we apply
inclusion--exclusion. The sets $A,B,C,M$ determine which tests occur
in each grouped term.

For full assignments $x,u\in\Omega$, let $q_{A,B}(x,u)$ indicate that
every constraint in $A\setminus\{c\}$ is violated by at least one of
$x,u$, and both satisfy $B\cup M$. There is no test on $c$ or $C$ in
this indicator. For $x,y\in\Omega$, assign $W_{A,B}(x,y)=h(A,B)$
exactly when they agree outside $\Lambda$, $q_{A,B}(x,y)=1$, and the
statuses of $c$ are
\begin{equation}\label{eq:orientation}
 c:\quad
 \begin{cases}
 \text{violated at }x,\ \text{satisfied at }y,& |A|\text{ odd},\\
 \text{satisfied at }x,\ \text{violated at }y,& |A|\text{ even}.
 \end{cases}
\end{equation}
All other weights are zero, and $W=\sum_{(A,B)\in\Pairs(J,c)}W_{A,B}$.
The  orientation of $c$ supplies the sign of the grouped
inclusion--exclusion term. We now verify the
combined balance of these mini-flows.

\begin{proposition}\label{prop:balance}
The flow $W=\sum_{(A,B)\in\Pairs(J,c)}W_{A,B}$ satisfies~\eqref{eq:balance}.
\end{proposition}
\begin{proof}
In this proof, sums over pairs range over the set $\Pairs(J,c)$ of generated pairs $(A,B)$, one for each generated skeleton $A$, and sums over assignments range over $\Omega$. Define $e(u)=\one_{\{u\not\models c\}}$
for every $u\in\Omega$, and fix $x\in\Omega$. The right side of~\eqref{eq:balance} equals
\begin{equation}\label{eq:targetsum}
 \sum_{u\in\Omega}(e(x)-e(u))
       \one_{\{x\models F\}}\one_{\{u\models F\}}.
\end{equation}
To check this, suppose first that $x\models F$. Among the $Z(F)$ assignments $u\models F$, exactly $Z(F)-Z(J)$ violate $c$. The sum is therefore $e(x)Z(F)-(Z(F)-Z(J))=Z(J)-(1-e(x))Z(F)$, as required. If $x$ violates $F$, every summand is zero. The factor $e(x)-e(u)$ records which endpoint violates the root; the other two indicators require both endpoints to satisfy the old constraints.

We recover this sum from local mutations in two steps. First, the weight $h(A,B)$ accounts for the possible outside  values of the second assignment. Second, inclusion--exclusion over the generated pairs restores the requirement that both assignments satisfy $F$. Thus it is enough to prove
\begin{align}
 \sum_y\bigl(W_{A,B}(x,y)-W_{A,B}(y,x)\bigr)
 &=(-1)^{|A|-1}\sum_u(e(x)-e(u))q_{A,B}(x,u),
       \label{eq:completiongoal}\\
 \sum_{(A,B)}(-1)^{|A|-1}q_{A,B}(x,u)
 &=\one_{\{x\models F\}}\one_{\{u\models F\}}.
       \label{eq:groupinggoal}
\end{align}
Summing~\eqref{eq:completiongoal} over the pairs and applying
\eqref{eq:groupinggoal} gives~\eqref{eq:targetsum}.

Fix a generated pair. For local values $\xi\in[D]^\Lambda$, put
$y_\xi=x[\Lambda\leftarrow\xi]$, the assignment equal to $\xi$ on
$\Lambda$ and to $x$ outside $\Lambda$. The orientation gives
\[
 W_{A,B}(x,y_\xi)-W_{A,B}(y_\xi,x)
 =(-1)^{|A|-1}h(A,B)(e(x)-e(y_\xi))q_{A,B}(x,y_\xi).
\]
For a violated-to-satisfied root, $e(x)-e(y_\xi)=1$. The right side is then $+h(A,B)$ for an odd skeleton, corresponding to an outgoing arc, and $-h(A,B)$ for an even skeleton, corresponding to an incoming arc, provided the remaining tests pass. Reversing the endpoints reverses the sign. It is the net flow that can be negative, never an arc weight.

If $x$ violates $M$, all relevant terms vanish. Otherwise, fix  $u|_\Lambda=\xi$. The status of $c$ and all tests on $A\cup B$ are now fixed. The only remaining restriction on the outside values of $u$ is satisfaction of $M$, with exactly $h(A,B)$ possible completions. Since $x$ satisfies $M$ and $y_\xi$ agrees with $x$ outside $\Lambda$, $y_\xi$ is one of these completions, and all have the same value of the local indicator. Hence
\[
 \sum_{u:u|_\Lambda=\xi}(e(x)-e(u))q_{A,B}(x,u)
 =h(A,B)(e(x)-e(y_\xi))q_{A,B}(x,y_\xi).
\]
Summing over $\xi$ proves~\eqref{eq:completiongoal}. For each permitted choice of local values there is only one actual destination $y_\xi$, since its outside values must remain those of $x$. Its weight $h(A,B)$ represents all $h(A,B)$ possible outside completions of the summation variable $u$, which give the same tests and sign. Thus the outside values are summed out.

It remains to prove~\eqref{eq:groupinggoal}. The underlying identity
is ordinary inclusion--exclusion, with its subsets grouped according
to the expansion procedure. For arbitrary numbers $(b_d)_{d\in F}$,
we claim that
\[
 \prod_{d\in F}(1-b_d)=
 \sum_{(A,B)}(-1)^{|A|-1}
 \prod_{a\in A\setminus\{c\}}b_a
 \prod_{d\in B\cup M}(1-b_d).
\]
Taking $b_d$ to be the indicator of the two-endpoint violation test
for $d$ gives~\eqref{eq:groupinggoal}. For the polynomial identity, expand the
left side by inclusion--exclusion and group its subsets $S\subseteq F$
according to the vulnerability rule. Once a generated skeleton $A$, and hence its associated pair $(A,B)$, is fixed, the set $S$ must contain the expansion constraints $A\setminus\{c\}$ and
omit the constraints at which expansion stopped, namely $C$.
Membership of the remaining constraints in $B\cup M$ is unrestricted.
We claim that these are exactly the subsets producing this pair;
in other words, every $S\subseteq F$ has exactly one representation
\begin{equation}\label{eq:subset}
 S=(A\setminus\{c\})\sqcup E,
 \qquad (A,B)\in\Pairs(J,c),\quad E\subseteq B\cup M(A,B).
\end{equation}
To see this, given $S$, run the vulnerability rule and expand through a selected
constraint exactly when it belongs to $S$. Then
$A\setminus\{c\}\subseteq S$.
Every member of $C$ enters at a branching step whose selected constraint
lies outside $S$. Since $C$ starts empty, $C\cap S=\varnothing$;
by~\eqref{eq:separation}, the remainder is in $B\cup M$.
Conversely, for any generated pair and any such $E$, membership in the
right side of~\eqref{eq:subset} reproduces every branching choice:
expansion constraints belong, while constraints put in $C$ do not.
Induction along the branch therefore recovers $(A,B)$, after which
$E=S\setminus(A\setminus\{c\})$. This proves uniqueness.

In the part indexed by $(A,B)$, the expansion sign is $(-1)^{|S|}=(-1)^{|A|-1+|E|}$. There are $|A|-1$ compulsory members, not $|A|$, because the inserted constraint $c$ does not belong to $F$. The contribution of this part is
\[
 \begin{aligned}
 &\sum_{E\subseteq B\cup M}(-1)^{|A|-1+|E|}
        \prod_{a\in A\setminus\{c\}}b_a\prod_{d\in E}b_d\\
 &\qquad=(-1)^{|A|-1}\prod_{a\in A\setminus\{c\}}b_a
                  \prod_{d\in B\cup M}(1-b_d).
 \end{aligned}
\]
A constraint in $B\cup M$ may be absent or present in $E$; summing these two choices produces $1-b_d$. A constraint in $C$ is absent from every subset in this part, so it supplies no factor, rather than a factor $1-b_d$. Its satisfaction is not required by this family; the combined sum over families restores all the old constraints. This proves the polynomial identity, then~\eqref{eq:groupinggoal}, and hence the proposition.
\end{proof}

This completes the ideal sampler. For every satisfiable $J$, the
flow above and Lemma~\ref{lem:absorption} transform $\mu_F$ into
$\mu_J$. Starting with $\mu_{\varnothing}$ and composing the
insertions gives $\mu_H$. No local-lemma hypothesis has been needed
for this conclusion. The remaining question is efficiency, since the
family of possible mutations of the ideal sampler can be exponentially large.

\section{An exponential tail for the insertion process}\label{sec:tail}

We will enumerate only mutations whose skeletons have logarithmic
size. To justify this, it is enough to show that the ideal sampler
rarely chooses a large mutation anywhere in its execution.
Equation~\eqref{eq:usage} reduces this issue to a sum of flow
weights. Throughout this section assume $8pQ\le1$.  Recall that $\Pairs(J,c)$ is the set of generated pairs $(A,B)$, indexed by their skeletons $A$, and $W_{A,B}$ is the corresponding mini-flow. 
\begin{proposition}\label{prop:tail}
For every constraint insertion $F\to J=F\cup\{c\}$ and integer $\ell\ge0$,
\begin{equation}\label{eq:tail}
 \frac1{Z(F)Z(J)}
 \sum_{\substack{(A,B)\in\Pairs(J,c)\\|A|>\ell}}
        \|W_{A,B}\|_1
 \le4p(23/25)^\ell.
\end{equation}
Consequently the complete exact process takes  at
most $m(1+4p)$ steps in expectation, and its probability of ever choosing $|A|>\ell$ is at
most $4mp(23/25)^\ell$.
\end{proposition}

By~\eqref{eq:usage}, the left side bounds the expected total number
of selections from the displayed families in one constraint insertion. It
therefore also bounds the chance of any such selection. At $\ell=0$
all families are included, giving at most $4p$ expected mutations;
there is also one stopping step per constraint insertion. Each insertion starts with the uniform law on the preceding instance, supplied by the exact output of the previous insertion. Summing its expected-use bound over the $m$ insertions gives the factor $m$ in the proposition; their independence is not needed. The probability of any large-family selection is at most this expected total number of selections.

There are two parts to the proof: Section~\ref{sec:singlepair}
bounds the normalized weight of one family, and
Section~\ref{sec:skeletonsum} bounds the number of skeletons.
Fix a constraint insertion and set
\begin{equation}\label{eq:numeric}
 t=\frac3{20Q},\qquad \kappa=(1-t)^{-2Q}.
\end{equation}
The single-family estimate, for $r=|A|$, is
\begin{equation}\label{eq:onepair}
 \frac{\|W_{A,B}\|_1}{Z(F)Z(J)}
 \le p\kappa(2p\kappa)^{r-1}.
\end{equation}
The number $n_r$ of generated pairs with $|A|=r$ satisfies
\begin{equation}\label{eq:treecount}
 \alpha_r=\frac{r^{r-1}}{r!},\qquad n_r\le\alpha_r\Delta_2^{r-1},
\end{equation}
where the coefficients obey
\begin{equation}\label{eq:treeboundary}
 \sum_{r\ge1}\alpha_r\mathrm e^{-(r-1)}\le\mathrm e.
\end{equation}
The decay from the probability bound beats the growth in the number
of skeletons if $\theta=2\mathrm e\Delta_2p\kappa<1$. For our choices, Appendix~\ref{app:tailconstants} verifies that
\begin{equation}\label{eq:theta}
 \theta<23/25,\qquad \mathrm e\kappa<4.
\end{equation}
These bounds give the desired conclusion immediately: the left side
of~\eqref{eq:tail} is at most
\[
 p\kappa\sum_{r>\ell}\alpha_r
                  \mathrm e^{-(r-1)}\theta^{r-1}
 \le\mathrm e p\kappa(23/25)^\ell
 \le4p(23/25)^\ell.
\]
It remains to prove the stated bounds.

\subsection{The single-pair bound}\label{sec:singlepair}

Fix a generated pair $(A,B)$ with $|A|=r$, and let $L_{A,B}$ be the number of ordered pairs of assignments to $\Lambda$ that pass the root orientation and the tests on $A\cup B$. Each such local pair extends to a full arc by choosing one common outside assignment satisfying $M$. There are $h(A,B)$ choices, and every resulting arc itself has weight $h(A,B)$. Therefore
\[
 \|W_{A,B}\|_1
 =L_{A,B}\,
   \underbrace{h(A,B)}_{\text{common outside choices}}\,
   \underbrace{h(A,B)}_{\text{weight of each arc}}.
\]
The two endpoints agree outside $\Lambda$: there are not two independent outside assignments. One factor of $h(A,B)$ counts arcs, and the other is their weight.

To bound $L_{A,B}$, take two independent uniform assignments on $\Lambda$. All $D^{2|\Lambda|}$ local pairs are equally likely. Write $p_a=\Pr[\neg a]\le p$ for the product-measure violation probability of $a\in A$. The prescribed change in the status of $c$ has probability $p_c(1-p_c)\le p$: a fixed family permits only one direction. For each other skeleton constraint the test passes unless both assignments satisfy it, so its probability is $1-(1-p_a)^2=2p_a-p_a^2\le2p$. The scopes in $A$ are disjoint, in both copies, so these probabilities multiply. Dropping the satisfaction tests on $B$ can only increase the number of accepted local pairs. Consequently,
\[
 L_{A,B}\le D^{2|\Lambda|}p(2p)^{r-1}.
\]
Multiplying by $h(A,B)^2$ cancels the local-volume factor, since $D^{2|\Lambda|}h(A,B)^2=Z(M)^2$. We obtain
\[
 \|W_{A,B}\|_1\le p(2p)^{r-1}Z(M)^2.
\]
After division by $Z(F)Z(J)$, the small factor $p(2p)^{r-1}$ is accompanied by $Z(M)^2/(Z(F)Z(J))$. This ratio is at least one: $M$ has fewer constraints and therefore at least as many solutions. To bound it, we first count how many constraints were deleted.

The two deleted sets are
\[
 F\setminus M=(A\setminus\{c\})\sqcup B\sqcup C,
 \qquad J\setminus M=A\sqcup B\sqcup C.
\]
Every member of $B$ is a neighbor of a member of $A$, and each member
of $C$ is adjacent to $A\cup B$. These inclusions hold in $G$ as well
as $G[J]$, so
\begin{equation}\label{eq:volumes}
 |A\cup B\cup C|\le Qr,
 \qquad |\Lambda|\le k(\Delta+1)r.
\end{equation}

It remains to bound the increase in the counts caused by these
deletions. We use the following form of the conditional local lemma,
stated in our CSP notation.

\begin{theorem}[Conditional local lemma {\cite[Theorem~2.1]{HSS}}]\label{thm:conditional-lll}
Let $X$ have the uniform product distribution on $[D]^V$, and suppose
there are numbers $\xi_a\in(0,1)$, $a\in H$, such that
\[
 \Pr[X\not\models a]
 \le \xi_a\prod_{b\in N_1(a)}(1-\xi_b)
 \qquad\text{for every }a\in H.
\]
Then $\Pr[X\models K]>0$ for every $K\subseteq H$. Moreover, for
every $K\subseteq H$ and every $c\in H\setminus K$,
\[
 \Pr[X\not\models c\mid X\models K]
 \le \Pr[X\not\models c]
       \prod_{b\in K\cap N_1(c)}(1-\xi_b)^{-1}
 \le \xi_c.
\]
All neighborhoods are taken in the original dependency graph $G$.
\end{theorem}

Deleting constraints only shrinks the neighborhoods, so the same
numbers $\xi_a$ satisfy the local-lemma hypothesis for every
subformula. Here we take $\xi_a=t=3/(20Q)$ for every $a\in H$.
Bernoulli's inequality and $\Delta\le Q-1$ give
\[
 (1-t)^{\Delta}\ge17/20,\qquad
 t(1-t)^{\Delta}\ge51/(400Q)>1/(8Q)\ge p.
\]
Since $|N_1(a)|\le\Delta$, this verifies the theorem's hypothesis:
$\Pr[X\not\models a]\le p\le t(1-t)^{\Delta}
\le t(1-t)^{|N_1(a)|}$.
The theorem therefore
implies that every subformula is satisfiable and that restoring a constraint eliminates at most a fraction $t$ of its solutions. For a subformula $K$ and a constraint $d\notin K$, the count identity behind this statement is
\[
 Z(K\cup\{d\})
 =Z(K)\Pr_{X\sim\mu_K}[X\models d]\ge(1-t)Z(K).
\]
At least a $1-t$ fraction of the current solutions survive each restoration. Apply the same conditional bound after each new constraint has been restored, and multiply these survival fractions; no independence between restorations is assumed. This gives
\begin{equation}\label{eq:conditional}
 \Pr[\neg c\mid K]\le t\quad(c\in H\setminus K),\qquad
 \frac{Z(K)}{Z(K')}\le(1-t)^{-|K'\setminus K|}
       \quad(K\subseteq K'\subseteq H).
\end{equation}
Here conditioning on $K$ means conditioning on all its constraints being satisfied. Put $s=|A\cup B\cup C|\le Qr$. From $M$ we must restore $s-1$ constraints to obtain $F$, and $s$ to obtain $J$, the extra constraint being $c$. Applying~\eqref{eq:conditional} to these two restoration sequences gives
\[
 \frac{Z(M)^2}{Z(F)Z(J)}
 \le (1-t)^{-(s-1)}(1-t)^{-s}
 \le (1-t)^{-2Qr}=\kappa^r.
\]
Combining this count-ratio bound with the local probability bound
proves~\eqref{eq:onepair}.

\subsection{Summing over skeletons}\label{sec:skeletonsum}

We fix the inserted constraint $c$. There are two facts to prove: the bound on the number of generated skeletons of each size, and the summability property of the coefficients in that bound.

An independent set connected in $G^2$ is connected in the graph joining
pairs at distance exactly two in $G$, of maximum degree $\Delta_2$.
Indeed, no two skeleton vertices are adjacent in $G$, so the edges connecting them in $G^2$ must all correspond to distance two. For an upper bound, we may count all connected $r$-vertex sets containing $c$ in this distance-two graph, even though not all of them are generated skeletons.

To prove~\eqref{eq:treecount} for $r\ge2$, choose a spanning tree for each set and label its vertices,
fixing root label $1$.
Choose just one spanning tree for each set, before assigning the labels. The constraint $c$ receives label $1$; the other $r-1$ labels can be assigned to the remaining constraints in any order.
Each set has $(r-1)!$ labelings.
We count these labeled realizations and then divide by $(r-1)!$. A realization records both a tree on the labels and the constraint carrying each label, so it determines the original vertex set. Thus different sets cannot contribute the same realization.

There are
$r^{r-2}$ labeled trees, and at most $\Delta_2^{r-1}$ vertex maps per
tree after fixing the root.
The first count is Cayley's formula with root label $1$ fixed. For the second, expose the images of the tree vertices starting from that root. Each child has at most $\Delta_2$ possible images adjacent to its parent's image, and there are $r-1$ children to place.
Allowing noninjective maps only increases the
bound.
Some counted maps identify distinct labels or have images that fail the independence or generation tests. We do not need to exclude them: every labeled realization of a generated skeleton is counted, and the extra maps only enlarge the upper bound. It follows that
\[
 n_r(r-1)!\le r^{r-2}\Delta_2^{r-1},
 \qquad
 \frac{r^{r-2}}{(r-1)!}=\frac{r^{r-1}}{r!}=\alpha_r.
\]
The case $r=1$ is immediate. A generated skeleton determines its
pair, so~\eqref{eq:treecount} also bounds the generated pairs.

We next show that the coefficients $\alpha_r$ remain summable after the exponential growth has been factored out. Keeping this coefficient information, rather than replacing it by a coarser exponential bound, is useful for the constant in the tail estimate.

For~\eqref{eq:treeboundary}, use the exponential generating functions $T_j$
of rooted labeled trees of height at most $j$:
$T_0(z)=z$ and $T_{j+1}(z)=z\exp(T_j(z))$.
Here height is the maximum number of edges from the root to a vertex, and the coefficient of $z^r$ is the number of such trees on $r$ labels divided by $r!$. Height zero permits only a single root, explaining $T_0(z)=z$. Removing the root of a tree of height at most $j+1$ leaves an unordered collection of rooted trees of height at most $j$. For $h$ branches, the contribution is $T_j(z)^h/h!$: products of exponential generating functions account for assigning disjoint labels to the branches, and the factor $1/h!$ removes their ordering. Summing over $h$ gives the exponential; multiplying by $z$ reinstates the root.

Their coefficients increase
to $\alpha_r$ in degree $r$, and $T_j(1/\mathrm e)\le1$ by induction.
For the coefficient assertion, the root may now carry any of the $r$ labels, unlike the fixed root label in the preceding map count. There are therefore $r\cdot r^{r-2}=r^{r-1}$ rooted labeled trees on $r\ge2$ vertices, giving coefficient $\alpha_r$; the singleton has coefficient one. Every such tree has height at most $r-1$, so all of them are included once $j\ge r-1$. For the bound at $1/\mathrm e$, the induction is
\[
 T_0(1/\mathrm e)=1/\mathrm e\le1,
 \qquad
 T_j(1/\mathrm e)\le1
 \ \Longrightarrow\ 
 T_{j+1}(1/\mathrm e)
   =\mathrm e^{-1}\exp(T_j(1/\mathrm e))\le1.
\]
This also justifies evaluating the limiting series at this point. For every integer $R\ge1$, the first $R$ coefficients have already reached their final values in $T_{R-1}$, and all coefficients are nonnegative. Hence
\[
 \sum_{r=1}^{R}\alpha_r\mathrm e^{-r}
 \le T_{R-1}(1/\mathrm e)\le1.
\]
Letting $R$ grow bounds the entire series by one; multiplying by $\mathrm e$ gives the normalization used in~\eqref{eq:treeboundary}.
Monotone convergence gives~\eqref{eq:treeboundary}.

In the summation at the beginning of Section~\ref{sec:tail}, these summable coefficients are multiplied by $\theta^{r-1}$. For $r>\ell$, this last factor is at most $(23/25)^\ell$. The remaining coefficient sum is at most $\mathrm e$, which gives the claimed tail without an additional geometric-series factor $1/(1-\theta)$.

This completes the proof of Proposition~\ref{prop:tail}.

\section{Approximate sampling}\label{sec:implementation}

We now implement the exact sampler with a finite computation. We
retain only small skeletons, use approximate  counts, and impose a
limit on the number of walk steps across all constraint insertions of a trial. Throughout this section assume
$8pQ\le1$ and $\Delta_2\ge2$.

A trial is a bounded execution of the sampler: it starts with a uniform product assignment, attempts all $m$ constraint insertions, and either returns a solution of $H$ or reports failure. We couple the trial to the ideal sampler and seek their first discrepancy along the ideal trajectory. After bounding the probability of that discrepancy, we remove the failure outcome using a bounded number of trials and, as a last resort, a fixed satisfying assignment.

\subsection{The truncated insertion process}\label{sec:trial}

Fix $0<\varepsilon\le1/2$, and set
\begin{equation}\label{eq:choices}
 L=\left\lceil\log_{25/23}\frac{16m}{\varepsilon}\right\rceil,
 \qquad \eta=\frac{\varepsilon}{32m}.
\end{equation}

These parameters address different costs. The cutoff $L$ limits the number and size of the mutation families we must generate; the accuracy $\eta$ limits the distortion of the choices within the retained families. The geometric tail makes logarithmic $L$ sufficient. We choose $\eta$ on the scale $\varepsilon/m$ because errors will be summed over the expected $O(m)$ steps of the ideal execution, rather than over the larger worst-case step cap imposed on the trial.

Retain only skeletons of size at most $L$, i.e., $|A|\le L$.
The only counts they need are $Z(F)$ for stopping and $Z(M)$ for
each retained mini-flow. Compute these to relative accuracy $\eta$
once, before running any trial. The mutation sets depend only on the
formula, so no new counting query is prompted by a sampled value.
All queries are unpinned, and all counts are positive by
\eqref{eq:conditional}. Round each estimate down to a multiple of $2^{-b}$, where $b=\lceil\log_2(1/\eta)\rceil$. Since every true count is at least one, rounding adds at most $\eta$ to its relative error. Reading and rounding an answer is included in its counting invocation. The stored estimates therefore satisfy
\begin{equation}\label{eq:rounded}
 (1-2\eta)Z(K)\le\widehat Z(K)\le(1+2\eta)Z(K).
\end{equation}
Use these same stored values in every trial, so repetition needs no further counting calls.

The precomputed objects and the state-dependent objects should be distinguished. For a fixed insertion and a fixed generated pair $(A,B)$, the sets $\Lambda,M$ and the count $Z(M)$ are already determined. The current assignment $x$ only determines which local replacements pass the arc tests. Finding those replacements requires predicate evaluations, not counts with the values of $x$ pinned. This is why every counting query can be fixed before the sampling randomness is drawn.

In the analysis we may regard the entire stored table as fixed, with any errors satisfying~\eqref{eq:rounded}. There is no assumption that its different errors cancel or are independent. The lower bound in~\eqref{eq:rounded} is positive, so an option has zero approximate weight exactly when it has zero exact weight. In particular, approximation changes the relative preferences among allowed moves, not which moves are allowed.

For a retained pair $(A,B)$ and a current assignment $x$, let
$\mathcal Y_{A,B}(x)$ be the possible destinations of its mini-flow.
They differ from $x$ only on $\Lambda$, so this set can be found by
local enumeration using the arc tests of Section~\ref{sec:flow}.
All its arcs have weight $h(A,B)$. The algorithm therefore chooses the family indexed by $(A,B)$
with weight $h(A,B)|\mathcal Y_{A,B}(x)|$, then chooses a uniform allowed
destination. For the exact walk this gives
\[
 \frac{h(A,B)|\mathcal Y_{A,B}(x)|}{d(x)}
 \frac1{|\mathcal Y_{A,B}(x)|}
 =\frac{W_{A,B}(x,y)}{d(x)}.
\]

The factor $|\mathcal Y_{A,B}(x)|$ is essential. A family with many allowed destinations carries the sum of all their arc weights, not the weight of just one arc. Choosing the family by this total and then choosing its destination uniformly recovers the prescribed probability of each labeled arc. If two families lead to the same assignment $y$, these are two ways to reach $y$, and their probabilities add. No disjointness of the destination sets is required.

Replacing each count by its stored approximation gives the following
weights for stopping and for the retained mutations:
\[
 \widehat w_{\mathrm{stop}}(x)
   =\widehat Z(F)\one_{\{x\models J\}},
 \qquad
 \widehat w_{A,B}(x)
   =D^{-|\Lambda|}\widehat Z(M)|\mathcal Y_{A,B}(x)|.
\]
An empty destination set has weight zero. As in the ideal walk, a step is either a mutation or a stopping choice. Set the limit for all insertions of one
trial together to
\begin{equation}\label{eq:decisioncap}
 T=\left\lceil\frac{16m}{\varepsilon}\right\rceil.
\end{equation}
A stopping choice ends the current constraint insertion, and after all $m$ insertions have stopped the trial returns its assignment, which satisfies $H$. The trial reports $\bot$ if taking another step would exceed $T$, or if the total available weight is zero.

The random choices reduce to uniform integer draws (Section~\ref{sec:resources}). Standard rejection sampling from a power-of-two range accepts each attempt with probability at least $1/2$. A trial makes at most $2T+1$ such draws, one for initialization and at most two per walk step. Cap each draw at $\lceil\log_2(16(2T+1)/\varepsilon)\rceil$ attempts, reporting $\bot$ if all fail. A union bound puts the probability of any such failure at most $\varepsilon/16$. Coupling capped draws to their uncapped versions leaves the choices identical unless this limit is reached.

\subsection{Error along the exact walk}\label{sec:trial-proof}

Let $\nu$ be the law of a trial, including possible output $\bot$,
and extend $\mu_H$ by $\mu_H(\bot)=0$.

The two laws now live on the same output space, consisting of satisfying assignments and $\bot$. Total variation accounts for both ways a trial can differ from the target: it may fail, or it may return solutions with slightly different probabilities. In particular, if $\delta=\TV(\nu,\mu_H)$, the event $\{\bot\}$ immediately gives $\nu(\bot)\le\delta$. We will still need a separate conditioning calculation to understand the output given success.

We prove
\begin{equation}\label{eq:trialbound}
 \TV(\nu,\mu_H)
 \le 4mp(23/25)^L+4\eta m(1+4p)
       +\frac{m(1+4p)}{T}+\frac{\varepsilon}{16}
 \le\frac{3\varepsilon}{8}.
\end{equation}
The four terms correspond to omitted mutations, approximate weights,
the step limit, and failure to produce a random draw within its allotted attempts.
Run the trial and the ideal sampler together, starting with the same
product assignment. As long as they agree, use a coupling of their
next choices that minimizes the chance they separate. If no
separation or cutoff occurs, their outputs agree. It is therefore
enough to bound the probability of the first such event.

Let $N$ be the total number of steps taken by the exact sampler across all $m$ constraint insertions, including the stopping step of each insertion.
Proposition~\ref{prop:tail} gives
\begin{equation}\label{eq:reference}
 \mathbb EN\le m(1+4p),\qquad
 \mathbb E[\text{selections with }|A|>L]
       \le4mp(23/25)^L.
\end{equation}
Let $\rho_i(x)$ be the probability that the exact process selects an omitted pair at state $x$ during the insertion of the $i$-th constraint. If the retained families and the stopping option have total exact weight zero, their approximate weights also sum to zero and the trial reports $\bot$. The exact process must then select an omitted family, so the probability of disagreement at this step is $1=\rho_i(x)$. Otherwise, conditioning the exact choice on a retained family or stopping changes the choice law by exactly $\rho_i(x)$ in total variation.

To see this directly, write $\rho=\rho_i(x)<1$. The omitted options have total probability $\rho$. Deleting them and renormalizing multiplies each retained probability by $1/(1-\rho)$, so the retained probabilities increase by a total of $\rho$. The sum of absolute changes is $2\rho$, and total variation is $\rho$.

This conditioning is a description of one choice at a fixed state, used only for the comparison. The algorithm need not compute $\rho$ or enumerate the omitted options. Nor are we claiming that the resulting walk is the ideal walk conditioned on never selecting a large family: that would condition on an entire future trajectory. Renormalizing the available options separately at each visited state is a different operation.

For a retained family $(A,B)$, replacing $Z(M)$ by $\widehat Z(M)$ multiplies every arc weight, and hence the total weight of the family, by $\widehat Z(M)/Z(M)$. The permitted destinations and their number are unchanged: they are determined by testing assignments, not by the counter. Thus the destination remains uniform within the chosen family. The stopping weight is multiplied by $\widehat Z(F)/Z(F)$. These factors can differ between families, but all lie in $[1-2\eta,1+2\eta]$ by~\eqref{eq:rounded}.

Index the retained families and the stopping option, when allowed, by $j$, writing $w_j$ for the exact weight  and $\widehat w_j=f_jw_j$ for the approximate weight, where $f_j\in[1-2\eta,1+2\eta]$. For the normalized laws on the retained options
$\pi_j=w_j/\sum_h w_h$ and
$\widehat\pi_j=\widehat w_j/\sum_h\widehat w_h$, put
$\bar f=\sum_j\pi_jf_j$. This is the ratio of the approximate total weight to the exact total weight, so $\widehat\pi_j=\pi_jf_j/\bar f$. Subtracting $\pi_j$ gives $\widehat\pi_j-\pi_j=\pi_j(f_j-\bar f)/\bar f$. Since both $f_j$ and their average $\bar f$ lie in $[1-2\eta,1+2\eta]$, we have $|f_j-\bar f|\le4\eta$ and $\bar f\ge1-2\eta$. Using also $\eta\le1/64$, we obtain
\begin{equation}\label{eq:one-step}
 \TV(\pi,\widehat\pi)
 =\frac1{2\bar f}\sum_j\pi_j|f_j-\bar f|
 \le\frac{2\eta}{1-2\eta}\le4\eta.
\end{equation}

There is no factor for the number of retained families in this estimate. The normalization error is their probability-weighted average $\bar f$, which stays in the same interval as the individual factors, no matter how many factors there are. The calculation permits all errors to point in an unfavorable direction. Thus we do not need to request proportionally greater accuracy when the number of generated families grows.

Given a common skeleton, the two processes use the same uniform
destination; a common stopping choice ends the same insertion.
Their one-step disagreement probability is therefore at most
$\rho_i(x)+4\eta$.

To sum the one-step errors, we use Proposition~\ref{prop:tail}, recalled in~\eqref{eq:reference}: it bounds the expected number of omitted selections during the complete ideal execution. This is exactly the expectation of the sum of the omission probabilities along that execution, even though any individual $\rho_i(x)$ can be large. Write $(I_j,X_j)$ for
the insertion and state before
exact step $j$, and let $O_j$ indicate selection of an omitted
pair, with $O_j=0$ after the exact process ends. At each existing step its conditional expectation is $\rho_{I_j}(X_j)$. Hence
\[
 \mathbb E\sum_{j<N}\rho_{I_j}(X_j)
 =\sum_{j\ge0}\mathbb E O_j
 =\mathbb E[\text{selections with }|A|>L].
\]
\par
Let $C_j$ be the event that the two executions still agree before
step $j$ and that no cutoff has yet occurred. A first choice
disagreement is possible only on this event, so its probability is at most
\[
 \mathbb E\sum_{j<N}\one_{C_j}
       \bigl(\rho_{I_j}(X_j)+4\eta\bigr).
\]
Dropping $\one_{C_j}$ only increases this nonnegative sum. We may
continue the ideal path with its prescribed random choices after the
comparison has stopped. Thus the reference path retains its
unconditional ideal law, and we obtain
\[
 \Pr[\text{a disagreement before a cutoff}]
 \le \mathbb E\sum_{j<N}\bigl(\rho_{I_j}(X_j)+4\eta\bigr)
 \le 4mp(23/25)^L+4\eta\mathbb EN.
\]

Conditioning on $C_j$ could bias the current assignment, but the
bound uses the complete ideal trajectory, not a trajectory under that
conditioning. In particular, we do not restart the proof at each
insertion assuming that the trial's current assignment is exactly
uniform: errors from earlier insertions are already included in the
first-discrepancy event. The accumulated count error is
$4\eta\mathbb EN$, not merely $4\eta T$. The step cap supplies a
worst-case clock, whereas the ideal expected length supplies the
sharper error budget.

The  step cutoff costs at most
$\Pr[N>T]\le\mathbb EN/T$.

Indeed, while choices and random draws agree, both processes have taken the same number of steps. If the ideal execution finishes within $T$ steps, a step-limit failure cannot be the first discrepancy. Thus the probability that the step cutoff is the first discrepancy is bounded using the ideal duration $N$, even though the implemented duration may have a different law after an earlier discrepancy.

Adding the probability, at most $\varepsilon/16$, that some random draw exhausts its allotted attempts proves the first inequality in
\eqref{eq:trialbound}. Substituting the chosen $L,\eta,T$ gives
$\varepsilon(1/4+p)\le3\varepsilon/8$, as $p\le1/8$.

The substitutions can be read one error at a time:
\[
\begin{array}{rcl}
 4mp(23/25)^L &\le& \varepsilon p/4,\\[2pt]
 4\eta m(1+4p) &=& \varepsilon(1+4p)/8,\\[2pt]
 m(1+4p)/T &\le& \varepsilon(1+4p)/16,\\[2pt]
 \text{random-draw failures} &\le& \varepsilon/16.
\end{array}
\]
The four displayed upper bounds sum to $\varepsilon(1/4+p)$. Thus we have both a small failure probability and a small unconditional output error, with slack left for conditioning on success and for returning an arbitrary satisfying assignment only as a last resort---the fallback assignment.

\subsection{Returning a satisfying assignment}\label{sec:trial-reduction}

A successful trial has a distribution close to $\mu_H$. To avoid
an unbounded wait for success, we allow only a fixed number of trials
and keep a satisfying assignment $x_*$ for the event that all of
them fail. Constructing the fallback assignment must have a guaranteed running-time bound: merely capping the randomized Moser--Tardos algorithm~\cite{MT} would leave a possibility of failure. Harris's deterministic LLL algorithm~\cite{Harris}
finds $x_*$ under our hypothesis in
$(n+m)^{O(k\log D)}D^{O(k)}$ time, without counting calls;
Appendix~\ref{app:fallback} gives the deduction. The probability
of using the fallback assignment will be small enough to preserve the sampling
guarantee.

Write $f=\nu(\bot)$ and let $\nu_+$ be the conditional law of a
successful trial. By~\eqref{eq:trialbound}, $f\le3\varepsilon/8<1/4$
and
\[
 \TV(\nu_+,\mu_H)
 \le\frac{\TV(\nu,\mu_H)}{1-f}
 \le\frac{6\varepsilon}{13}<\frac\varepsilon2.
\]

The conditioning step deserves its own calculation: successful trials need not be exactly uniform. Since $\nu(x)=(1-f)\nu_+(x)$ for a satisfying assignment $x$, and $\mu_H$ gives no mass to $\bot$, writing $\delta=\TV(\nu,\mu_H)$ gives
\[
 \sum_x|\nu(x)-\mu_H(x)|=2\delta-f.
\]
The sum here is over satisfying assignments only; the missing contribution $f$ is exactly the discrepancy at $\bot$. Consequently
\[
 \begin{aligned}
  2(1-f)\TV(\nu_+,\mu_H)
  &=\sum_x|\nu(x)-(1-f)\mu_H(x)|\\
  &\le\sum_x|\nu(x)-\mu_H(x)|+f\sum_x\mu_H(x)\\
  &=(2\delta-f)+f=2\delta.
 \end{aligned}
\]
Dividing by $2(1-f)$ proves the first inequality. For the numerical bound, $f\le3\varepsilon/8\le3/16$ gives $1-f\ge13/16$, and hence $\delta/(1-f)\le6\varepsilon/13$. Discarding failed trials therefore costs only this controlled renormalization; it does not make the remaining distribution exactly uniform.

Run at most
\begin{equation}\label{eq:tries}
 N_{\mathrm{try}}=\left\lceil\log_4\frac4\varepsilon\right\rceil
\end{equation}
independent trials and return the first successful output. If all
fail, return $x_*$. 

For completeness, returning the first success has exactly the same conditional law $\nu_+$ as returning one trial conditional on success. For any set $E$ of satisfying assignments and any $1\le t\le N_{\mathrm{try}}$,
\[
 \Pr[\text{first success at trial }t,\ \text{output}\in E]
     =f^{t-1}(1-f)\nu_+(E).
\]
Summing this geometric series over $t$ yields
$(1-f^{N_{\mathrm{try}}})\nu_+(E)$. Thus repetition changes the probability of producing a successful trial, but does not favor one successful output over another. The cached count estimates are the same in every trial; their errors need not be resampled. Once that table is fixed, fresh sampling randomness makes the trials independent with the same law $\nu$.

The probability that all trials fail is
$f^{N_{\mathrm{try}}}\le\varepsilon/4$, so the full output law
\[
 (1-f^{N_{\mathrm{try}}})\nu_+
       +f^{N_{\mathrm{try}}}\delta_{x_*}
\]
is within $\varepsilon/2+\varepsilon/4<\varepsilon$ of $\mu_H$.
Both parts of this mixture are supported on satisfying assignments.
The fallback assignment need not be random: its small probability of use is
already included in the error. All trials use the same precomputed
counts and fresh sampling randomness.

The final guarantee has two separate sources. Every possible output is a solution because a successful trial returns a solution and the fallback assignment is also a solution. Approximate uniformity follows because the successful part is close to uniform and the fallback assignment receives very little mass:
\[
 \TV\!\left((1-f^{N_{\mathrm{try}}})\nu_+
       +f^{N_{\mathrm{try}}}\delta_{x_*},\mu_H\right)
 \le (1-f^{N_{\mathrm{try}}})\TV(\nu_+,\mu_H)
       +f^{N_{\mathrm{try}}}.
\]
No typicality property of $x_*$ is required. Its construction is nevertheless part of the running time: the implementation computes it once, and pays that cost whether or not its value is eventually returned.

\subsection{Running time}\label{sec:resources}

The work has three parts: generate the retained skeletons, obtain their
counts once, and enumerate their permitted mutations when making a step. Put
\[
 P=\sum_{r=1}^{L}(4\Delta_2)^{r-1}.
\]
A retained skeleton is a connected set of at most $L$ vertices in
the exact-distance-two graph, contains the inserted constraint, and
determines its pair by Lemma~\ref{lem:structure}. There are at most
$P$ retained pairs per constraint insertion, and all can be enumerated in
$\operatorname{poly}(n,m,L)P$ time. Appendix~\ref{app:enumeration}
gives the enumeration.

There is one prefix count per constraint insertion and one exterior count per
retained pair, for a total of
\begin{equation}\label{eq:querycount}
 m(1+P)\le m(4\Delta_2)^L
       =(m/\varepsilon)^{O(\log\Delta_2)}
\end{equation}
invocations at relative accuracy $\eta=\varepsilon/(32m)$. The
queries are unchanged by repetition of the trial.

By~\eqref{eq:volumes}, each retained mutation involves at most
$K_{\rm var}=\min\{n,k(\Delta+1)L\}$ variables. Thus one walk step
requires at most $PD^{K_{\rm var}}$ local assignments to be tested,
with polynomial overhead for each. There are at most
$N_{\mathrm{try}}T$ steps in all.

Two different enumerations appear here. Generating a skeleton follows bounded-degree connections in the dependency graph and costs roughly $(4\Delta_2)^L$, rather than inspecting every $L$-element subset of the $m$ constraints. Once that skeleton is fixed, generating its destinations varies only the at most $K_{\rm var}$ mutable variables, rather than all $n$ variables. Since $L=O(\log(m/\varepsilon))$,
\[
 \begin{aligned}
 P&\le(m/\varepsilon)^{O(\log\Delta_2)},\\
 D^{K_{\rm var}}&\le(m/\varepsilon)^{O(k(\Delta+1)\log D)}.
 \end{aligned}
\]
These bounds explain the parameter-dependent exponent in the theorem. They give polynomial work in $n,m,1/\varepsilon$ for fixed structural parameters; they do not assert a running time polynomial in all structural parameters simultaneously. The bound of $N_{\mathrm{try}}T$ now makes this per-step cost a worst-case bound, rather than an expected one.

The stored counts have $O(n\log D+\log(1/\eta))$ bits. Multiplying all weights at a step by $2^bD^n$ makes them integers without changing the choice probabilities. The sum of these integers has $O(n\log D+\log(1/\eta)+\log(P+1))$ bits, so a proportional choice is a uniform integer draw into consecutive intervals of those lengths. Uniform destination and initial-assignment draws also have ranges of polynomial bit length. Thus arithmetic and the logarithmic retry cap add only polynomial overhead. Including the one deterministic
search, the total number of additional bit operations and predicate evaluations is at most
\begin{equation}\label{eq:work}
 \operatorname{poly}(n,m,k,L,\log D,1/\varepsilon)\,P D^{K_{\rm var}}
       +(n+m)^{O(k\log D)}D^{O(k)}.
\end{equation}
Since
$L=O(\log(m/\varepsilon))$ and
$\Delta_2\le\Delta(\Delta-1)$, this is
$(n+m/\varepsilon)^{O(k\Delta\log D)}$. Together with
\eqref{eq:querycount} and the output bound, this proves
Theorem~\ref{thm:oracle}.

\section{Asymmetric sampling}\label{sec:asymmetric}

We prove Theorem~\ref{thm:asymmetric}. Only the tail estimate and the
construction of the fallback assignment need to change. The insertion
flow is unchanged, and all neighborhoods below are in the ambient graph
$G$.

\paragraph{The family weights.}
Fix witnesses $x_c$ as in the theorem and put $t_c=Cx_c$. Since
\[
 p_c\le Cx_c\prod_{b\in\mathcal B(c)}(1-x_b)
 \le t_c\prod_{b\in N_1(c)}(1-t_b),
\]
Theorem~\ref{thm:conditional-lll} gives positivity of every count and
$Z(K\cup\{b\})\ge(1-t_b)Z(K)$ for $b\notin K$. Define
\[
 \kappa_a=\prod_{b\in\mathcal B(a)}(1-Cx_b)^{-2},
 \qquad w_a=2p_a\kappa_a.
\]
The deleted-region inclusion from~\eqref{eq:volumes} is
$J\setminus M\subseteq\bigcup_{a\in A}\mathcal B(a)$.
Restoring these constraints as in Section~\ref{sec:singlepair} bounds
$Z(M)^2/[Z(F)Z(J)]$ by $\prod_{a\in A}\kappa_a$; overlaps only
increase this upper bound. Keeping the individual probabilities in the
same local endpoint calculation gives
\[
 \frac{\|W_{A,B}\|_1}{Z(F)Z(J)}
 \le p_c\kappa_c\prod_{a\in A\setminus\{c\}}2p_a\kappa_a
 =\frac12\prod_{a\in A}w_a.
\]

\paragraph{A weighted tree comparison.}
We use a weighted version of the enumeration in
Section~\ref{sec:skeletonsum}. If nonnegative vectors $\omega,v$
satisfy
\[
 \omega_a\prod_{b\in N_2(a)}(1+v_b)\le v_a\qquad(a\in H),
\]
then
\begin{equation}\label{eq:asym-tree-comparison}
 \sum_{(A,B)\in\Pairs(J,c)}\prod_{a\in A}\omega_a\le v_c
 \qquad(c\in J\subseteq H).
\end{equation}
To see this, choose a rooted spanning tree of each generated skeleton in
the ambient exact-distance-two graph, using fixed tie-breaking. Its
label set recovers the skeleton and hence the pair. Enlarge this tree
family by allowing labels to repeat at different depths or in different
branches, but at most one child of each type $b\in N_2(a)$ at a node
of type $a$. The total weight of such trees of height at most $h$ obeys
\[
 T_a^{(-1)}=0,\qquad
 T_a^{(h+1)}=\omega_a\prod_{b\in N_2(a)}(1+T_b^{(h)})\le v_a
\]
by induction. Increasing $h$ proves~\eqref{eq:asym-tree-comparison}.
The comparison holds for every deletion subinstance using the same
ambient neighborhoods, and will also be used in the counting proof.

\paragraph{The asymmetric tail.}
Set $\beta=1-2C$, $y_a=\beta x_a$, and $z=\beta/(2C)>1$.
The identity
\[
 (1-Cu)^2-(1-u)(1+\beta u)=(1-C)^2u^2\ge0
 \qquad(0\le u\le1)
\]
pays for both the deletion factors and the distance-two branches:
applying it to each factor in the product hypothesis gives
\[
 p_a\kappa_a\prod_{b\in N_2(a)}(1+y_b)\le Cx_a,
 \qquad zw_a\prod_{b\in N_2(a)}(1+y_b)\le y_a.
\]
For factors outside $N_2(a)$, the factor $1+\beta u$ is unnecessary.
Apply~\eqref{eq:asym-tree-comparison} with $\omega_a=zw_a$ and
$v_a=y_a$, and divide by $2z$:
\[
 \sum_{(A,B)\in\Pairs(J,c)}z^{|A|-1}
       \frac{\|W_{A,B}\|_1}{Z(F)Z(J)}
 \le\frac{y_c}{2z}=Cx_c.
\]
The mass restricted to $|A|>\ell$ is consequently at most
$Cx_cz^{-\ell}$. By~\eqref{eq:usage}, the complete ideal execution has
at most $(1+C)m$ expected steps and at most $Cmz^{-L}$ expected
selections of omitted families.

\paragraph{Implementation and the fallback assignment.}
Keep $\eta,T$ and the implementation of Section~\ref{sec:implementation},
changing only the skeleton cutoff to
$L=\lceil\log_z(16m/\varepsilon)\rceil$.
The proof of~\eqref{eq:trialbound} now gives
\[
 \TV(\nu,\mu_H)
 \le Cmz^{-L}+4\eta(1+C)m+\frac{(1+C)m}{T}
               +\frac{\varepsilon}{16}
 \le\frac{1+C}{4}\varepsilon<\frac{3\varepsilon}{8}.
\]
Thus conditioning, repetition, and enumeration give the claimed bounds
with $L=O_C(\log(m/\varepsilon))$, once the fallback assignment is
available. For $C=23/96$, even the original cutoff is unchanged.

For its construction, set $S=\{a\}\cup N_1(a)$. Because
$S\subseteq\mathcal B(b)$ for every $b\in S$, the product hypothesis
implies, with $s=\sum_{b\in S}x_b$,
\[
 \sum_{b\in S}p_b
 \le Cs\prod_{d\in S}(1-x_d)
 \le Cs\mathrm e^{-s}\le C/\mathrm e.
\]
Also $p_b\le C/4$. Use $a_{\rm H}$ from Appendix~\ref{app:fallback},
so $q_b=p_b^{1-a_{\rm H}}\le2^{1/4}p_b$. The witnesses $u_b=2p_b$
satisfy
\[
 u_b\prod_{d\in N_1(b)}(1-u_d)
 \ge2p_b(1-2C/\mathrm e)>2^{1/4}p_b\ge q_b.
\]
Harris's asymmetric bound~\cite[Proposition~2.3(3)]{Harris} gives
\[
 W_{a_{\rm H}}\le\sum_b\frac{u_b}{(1-u_b)q_b}<\frac{16m}{7}.
\]
Here $W_{a_{\rm H}}$ is the same work parameter used in
Appendix~\ref{app:fallback}; the summand bound is the asymmetric
counterpart of the symmetric bound $\mathrm e$ used there. Indeed,
$u_b<1/8$ and $q_b\ge p_b$. The same runtime
formula~\cite[Theorem~3.5]{Harris} and local-table calculation give
$(n+m)^{O(k\log D)}D^{O(k)}$ time, without counting calls. This proves
Theorem~\ref{thm:asymmetric}.

\section{Counting for the asymmetric sampler}\label{sec:asymmetric-counting}

We prove Theorem~\ref{thm:asymmetric-counting} by adapting the counting
proof of Liu et al.~\cite{CountingLLL}. In particular, we use the
unconditional-coefficient marginal expansion in their Lemma~5.1 and
the size-budget recursion of their Section~4. We write the expansion
using our generated pairs by taking the diagonal of
\eqref{eq:groupinggoal}; the weighted estimates then extend their
stability argument to the product condition. This section supplies the
deletion counts needed by the asymmetric sampler where their stated
symmetric counting theorem does not apply.

\subsection{The marginal recursion}

For $c\in J\subseteq H$, put $F=J\setminus\{c\}$ and
$\beta_{J,c}=1-Z(J)/Z(F)$. For $(A,B)\in\Pairs(J,c)$, use the existing
$\Lambda,M$, and define the local coefficient
\[
 \lambda_{A,B}=\Pr_{\xi\sim\mathrm{Unif}([D]^\Lambda)}
 [\xi\text{ violates every constraint in }A
             \text{ and satisfies every constraint in }B]
 \le\prod_{a\in A}p_a.
\]
Set $u=x$ in~\eqref{eq:groupinggoal}, multiply by
$\one_{\{x\not\models c\}}$, and average over the product distribution.
The left-hand expectation is $(Z(F)-Z(J))/D^n$. By
\eqref{eq:separation}, the local tests on $\Lambda$ are independent
of satisfaction of $M$, so the term for $(A,B)$ has expectation
$\lambda_{A,B}Z(M)/D^n$, with its prescribed sign. Dividing by
$Z(F)/D^n$ gives
\[
 \beta_{J,c}=\sum_{(A,B)\in\Pairs(J,c)}
       (-1)^{|A|-1}\lambda_{A,B}\frac{Z(M)}{Z(F)}.
\]
For this term, list $F\setminus M=\{d_1,\ldots,d_s\}$ in the fixed order
and put $J_i=M\cup\{d_1,\ldots,d_i\}$. Telescoping the ratio gives
\begin{equation}\label{eq:asym-count-recursion}
 \beta_{J,c}=\sum_{(A,B)\in\Pairs(J,c)}
 (-1)^{|A|-1}\lambda_{A,B}
       \prod_{i=1}^{s}(1-\beta_{J_i,d_i})^{-1}.
\end{equation}
Every $J_i\subseteq F$ is a proper subformula of $J$.
The coefficient requires only one local assignment enumeration, not a
counting call.

We will keep each recursive estimate in an interval known to contain
the true marginal, so we need a computable upper bound. The witnesses
give the bound $t_c=Cx_c$, as in Section~\ref{sec:asymmetric}, but the
algorithm does not know them.
Compute the $p_c$ from their local tables and set
\[
 \overline t_c=p_c\prod_{b\in N_1(c)}(1-2p_b)^{-1}.
\]
The first-neighborhood estimate in Section~\ref{sec:asymmetric} gives
$\sum_{b\in N_1(c)}p_b\le C/\mathrm e$, and hence
$\overline t_c\le p_c/(1-2C/\mathrm e)<2p_c$. Consequently
\[
 \overline t_c\prod_{b\in N_1(c)}(1-\overline t_b)
 \ge\overline t_c\prod_{b\in N_1(c)}(1-2p_b)=p_c.
\]
Thus $\overline t$ is itself a conditional-LLL witness. Moreover,
$2p_b\le2Cx_b<x_b$, so the product hypothesis also gives
\[
 \overline t_c
 \le Cx_c\frac{\prod_{b\in\mathcal B(c)}(1-x_b)}
                   {\prod_{b\in N_1(c)}(1-2p_b)}
 \le t_c.
\]
By Theorem~\ref{thm:conditional-lll}, uniformly over rooted subformulas,
\begin{equation}\label{eq:asym-count-interval}
 0\le\beta_{J,c}\le\overline t_c\le t_c<1/4.
\end{equation}
Only $\overline t_c$, not $t_c$, is used by the algorithm.

\subsection{A weighted stability bound}

Put $q=3/4$, $\rho=q^{-1}=4/3$, and $a_0=8/3$. For a family $P=(A,B)$,
write its summand in~\eqref{eq:asym-count-recursion} as
\[
 f_P(v)=(-1)^{|A|-1}\lambda_{A,B}\prod_{i=1}^s(1-v_i)^{-1},
 \qquad
 H_P=\lambda_{A,B}\prod_{i=1}^s(1-a_0t_{d_i})^{-1}.
\]
All the inflated denominators are positive. For every child vector
$0\le v_i\le\overline t_{d_i}$, we have $|f_P(v)|\le H_P$ and
\begin{equation}\label{eq:asym-count-derivative}
 \sum_{i=1}^s t_{d_i}|\partial_i f_P(v)|\le\frac{H_P}{a_0-1}
 =\frac35 H_P.
\end{equation}
Indeed, $\partial_i f_P(v)=f_P(v)/(1-v_i)$, so
$|\partial_i f_P(v)|$ is nondecreasing in every nonnegative coordinate.
Since $\overline t\le t$, the left side is at most its value at
$(t_{d_i})_i$, namely
$\lambda_{A,B}\prod_i(1-t_{d_i})^{-1}
\sum_i t_{d_i}/(1-t_{d_i})$. Expanding the nonnegative product shows
\[
 \prod_i\frac{1-t_{d_i}}{1-a_0t_{d_i}}
 =\prod_i\left(1+\frac{(a_0-1)t_{d_i}}{1-a_0t_{d_i}}\right)
 \ge1+(a_0-1)\sum_i\frac{t_{d_i}}{1-t_{d_i}}.
\]
This proves~\eqref{eq:asym-count-derivative}, also when the coefficient
or the sum is zero. The extra room in the inflated factors pays for
all the child derivatives, rather than just the size of the term.

The same deleted-region inclusion as in Section~\ref{sec:asymmetric}
gives $H_P\le\prod_{b\in A}g_b$, where
\[
 g_c=p_c\prod_{b\in\mathcal B(c)}(1-a_0t_b)^{-1}.
\]
The elementary identity
\[
 (1-a_0Cu)-(1-u)(1+\rho Cu)
 =(1-4C)u+\rho Cu^2\ge0\qquad(0\le u\le1)
\]
and the product hypothesis yield
$\rho g_c\prod_{b\in N_2(c)}(1+\rho t_b)\le\rho t_c$.
Apply~\eqref{eq:asym-tree-comparison} with node weights $\rho g_c$
and bounds $\rho t_c$. For every rooted subformula,
\begin{equation}\label{eq:asym-count-envelope}
 \sum_{P\in\Pairs(J,c)}q^{-|A(P)|}H_P\le\frac43t_c,
 \qquad
 \sum_{P\in\Pairs(J,c)}q^{-|A(P)|}
       \sup_v\sum_i t_{d_i}|\partial_i f_P(v)|\le\frac45t_c.
\end{equation}
The supremum in each term is over its own child box. The first bound
controls omitted terms. The second is a contraction when the error
of a child rooted at $d$ is measured in units of $t_d$; an unweighted
maximum error would discard this useful dependence on the root.

\subsection{The bounded recursion and its accuracy}

Set $\widehat\beta_0(J,c)=0$. At integer budget $\ell\ge1$, evaluate
\eqref{eq:asym-count-recursion} using only families of size
$r=|A|\le\ell$, with child value
$\widehat\beta_{\ell-r}(J_i,d_i)$. Sum the signed terms exactly and clip
the result to $[0,\overline t_c]$: replace a value $v$ by
$\min\{\overline t_c,\max\{0,v\}\}$, leaving values inside the interval
unchanged and moving values outside it to the nearest endpoint.
Round the clipped value down to the grid with spacing $p_c2^{-b_\ell}$,
where $b_\ell$ is the least integer satisfying
$2^{-b_\ell}\le q^\ell/8$. This computation uses only the known
probabilities, generated pairs, and local coefficients.

We claim, simultaneously for all rooted subformulas, that
\begin{equation}\label{eq:asym-count-error}
 |\widehat\beta_\ell(J,c)-\beta_{J,c}|\le8t_cq^\ell.
\end{equation}
At budget zero this follows from~\eqref{eq:asym-count-interval}.
For a retained size-$r$ family, the induction hypothesis and the
mean-value formula bound its error by $8q^{\ell-r}$ times the
weighted derivative in~\eqref{eq:asym-count-derivative}. The segment
between the two child vectors stays in the prescribed box. The omitted
terms have absolute sum at most $\sum_{P:\,|A(P)|>\ell}H_P$.
Since the true marginal lies in $[0,\overline t_c]$, clipping cannot
increase the error. Rounding contributes at most
$p_cq^\ell/8\le t_cq^\ell/8$. By~\eqref{eq:asym-count-envelope},
\[
 |\widehat\beta_\ell(J,c)-\beta_{J,c}|
 \le t_cq^\ell\left(\frac18+\frac43+8\cdot\frac45\right)
 <8t_cq^\ell.
\]
This accounts for rounding, omitted terms, and every retained child.
The constants in this estimate are uniform over $0<C<1/4$.

For an arbitrary $K\subseteq H$, list its members as
$d_1,\ldots,d_j$ in the inherited order and put
$K_i=\{d_1,\ldots,d_i\}$. Return
\[
 L_{\rm cnt}=\left\lceil\log_{4/3}\frac{16m}{\delta}\right\rceil,
 \qquad
 \widehat Z(K)=D^n\prod_{i=1}^j
       \bigl(1-\widehat\beta_{L_{\rm cnt}}(K_i,d_i)\bigr).
\]
The exact product is $Z(K)$. Every exact and estimated satisfaction
factor is greater than $3/4$, so
\[
 \left|\log\frac{\widehat Z(K)}{Z(K)}\right|
 \le\frac43\sum_{i=1}^j8t_{d_i}q^{L_{\rm cnt}}
 <\frac83m q^{L_{\rm cnt}}\le\frac\delta6.
\]
For $0<\delta\le1/2$ this implies the stated relative error. 

\subsection{Running time and the symmetric corollary}

All graph bounds remain those of $H$. At a recursive call, the
size-$r$ candidate count is at most $(4\Delta_2)^{r-1}$ by
Appendix~\ref{app:enumeration}; evaluating $\lambda_{A,B}$ takes at
most $D^{k(\Delta+1)r}$ local assignments; and the term makes at most
$Qr$ child calls, each at budget $\ell-r$. These bounds give total
work exponential in the budget, with base
\[
 (Q+1)(\Delta_2+1)^{O(1)}D^{O(k(\Delta+1))},
\]
apart from an input-polynomial factor. To see the recursion bound,
put $A_*=4\Delta_2$. For $R>A_*$, the normalized cost of
recursive children satisfies
\[
 \sum_{r\ge1}QrA_*^{r-1}R^{-r}
 =\frac{Q/R}{(1-A_*/R)^2}
 \le\frac29<\frac12
 \qquad\bigl(R\ge\max\{4A_*,8Q\}\bigr).
\]
Choose $R$ also at least a sufficiently large constant multiple of the
local-work base above. Induction then bounds the call cost by an
input-polynomial factor times $R^\ell$.

The same argument includes rational arithmetic. Each returned value
is a multiple of $p_c2^{-b_\ell}$ between zero and $2p_c$, and therefore
has $O(k\log D+\ell)$ bits. A size-$r$ term contains at most $Qr$
reciprocal factors. The denominator of an exact sum can have bit length
proportional to the number of terms times their individual bit bound.
Summing all terms exactly still costs a fixed polynomial in their
number and bit lengths, absorbed by the constant power of
$\Delta_2+1$ above; only the rounded value is returned to the parent.
Since $L_{\rm cnt}=O(\log(m/\delta))$,
$Q\le1+\Delta^2$, and $\Delta_2\le\Delta(\Delta-1)$, the $j$ marginal
evaluations and their final product have the claimed complexity.
All constants in these work and error bounds are absolute; they do not
depend on $1/4-C$. This proves Theorem~\ref{thm:asymmetric-counting}.

The same counter applies to every deletion query. At accuracy
$\varepsilon/(32m)$, it supplies the oracle of
Theorem~\ref{thm:asymmetric}, giving total sampling time
$(n+m/\varepsilon)^{O_C(k\Delta\log D)}$ under the product condition.

It remains to prove Theorem~\ref{thm:symmetric-complete}. The maximum of
$u(1-u)^Q$ for $0<u<1$ is
\[
 b_Q:=\frac{Q^Q}{(Q+1)^{Q+1}},
 \qquad\text{attained at }u=\frac1{Q+1}.
\]
Under its hypothesis $p<b_Q/4$, choose any $C$ with
$p/b_Q<C<1/4$ and set $x_c=1/(Q+1)$ for every $c$. Since
$|\mathcal B(c)|\le Q$, these witnesses satisfy the product condition.
Theorem~\ref{thm:asymmetric-counting} therefore gives counting for every
deletion subinstance with its uniform running-time bound.

For sampling, use Theorem~\ref{thm:oracle}, rather than
Theorem~\ref{thm:asymmetric}, because
\[
 8pQ<2\left(\frac{Q}{Q+1}\right)^{Q+1}<\frac2{\mathrm e}<1.
\]
Its cutoff and complexity do not depend on the chosen $C$. Substituting
the new counter at accuracy $\varepsilon/(32m)$ into its counting-call
bound gives the total time claimed in
Theorem~\ref{thm:symmetric-complete}. This argument uses only
$C<1/4$, even arbitrarily close to the boundary.

For the simpler sufficient condition stated in the introduction, note
that
\[
 \frac{Q^Q}{4(Q+1)^{Q+1}}
 =\frac1{4(Q+1)}\left(\frac Q{Q+1}\right)^Q
 >\frac1{4\mathrm e(Q+1)}.
\]
Thus $4\mathrm e p(Q+1)\le1$ suffices. Also
$Q+1\le\Delta^2+2\le(\Delta+1)^2$, so the previously available
symmetric regime is included. When $Q\sim\Delta^2$, the leading
constant is still $4\mathrm e$; this corollary refines the neighborhood
bound, rather than the worst-case leading constant of
\cite{CountingLLL}.

\clearpage
\normalsize
\appendix
\section{Supporting estimates}\label{app:implementation}

The remaining details are the numerical bounds used in the tail estimate, generation of the retained mutations, and construction of the fallback assignment.

\subsection{Numerical estimates for the tail bound}\label{app:tailconstants}

We verify~\eqref{eq:theta}. Write $a_0=3/20$, so $t=a_0/Q$. Bernoulli's inequality gives $\kappa\le(20/17)^2$, and $\mathrm e<11/4$ gives $\mathrm e\kappa<4$. For the first inequality, use $\Delta_2\le Q-1$ and, for $0<u<1$, $\log(1-u)\le-u$ and $-\log(1-u)\le u/(1-u)$:
\[
 \log\left(\frac{\Delta_2}{Q}\kappa\right)
 \le-\frac1Q+\frac{2a_0}{1-a_0/Q}
 =2a_0+\frac1Q\left(\frac{2a_0^2}{1-a_0/Q}-1\right)\le2a_0.
\]
The last bracket is negative because $2a_0^2/(1-a_0/Q)\le9/170<1$. Therefore $\theta\le\mathrm e^{13/10}/4<23/25$. For example,
\[
 \mathrm e^{13/10}\le
 \sum_{j=0}^{6}\frac{(13/10)^j}{j!}
 +\frac{(13/10)^7}{7!}\frac1{1-13/80}<\frac{92}{25}.
\]

\subsection{Enumerating retained mutations}\label{app:enumeration}

To generate the retained pairs without running the full branching
procedure, enumerate rooted ordered trees of size $r\le L$ and their maps into the graph
joining constraints at distance exactly two in $G$, fixing the root at the inserted
constraint. There are at most $4^{r-1}\Delta_2^{r-1}$ such
maps. Keep distinct injective images inside the prefix that are
independent in $G$. For each image $A$, replay the vulnerability
rule, expanding exactly at members of $A$, and keep the pair if the
expansion set is precisely $A$. Lemma~\ref{lem:structure} shows that
this generates all retained pairs in $\operatorname{poly}(n,m,L)P$ time.

\subsection{Constructing the fallback assignment}\label{app:fallback}

Under $8pQ\le1$, the algorithm computes $x_*$ once, before the trials, using the symmetric consequence
of~\cite[Proposition~2.3(1), Theorem~3.5]{Harris} stated below.
Let $p_c$ be the violation probability of constraint $c$ under independent
uniform variable values, and let $\mathsf T_{\rm loc}$ bound the time
needed to compute that probability after fixing some variables in its
scope. If $0<a_{\rm H}<1/2$ and
$\mathrm e(\Delta+1)\max_c p_c^{1-a_{\rm H}}\le1$,
Harris's algorithm finds a satisfying assignment in
$(n+m)^{O(1/a_{\rm H})}D\,\mathsf T_{\rm loc}$ time.
Indeed, applying the cited proposition to the vector
$(p_c^{1-a_{\rm H}})_c$ gives $W_{a_{\rm H}}\le\mathrm e m$ in the
notation of the cited algorithmic theorem. Substitution in its runtime
$W_{a_{\rm H}}^{O(1/a_{\rm H})}nD\,\mathsf T_{\rm loc}$ yields the
stated bound.

Compute the $p_c$ by enumerating local tables; they are positive by the nontriviality assumption. Take
$a_{\rm H}=1/(4k\lceil\log_2D\rceil)$. Each $p_c$ is at least
$D^{-k}$, so
\[
 p_c^{1-a_{\rm H}}\le2^{1/4}p_c,
 \qquad
 \mathrm e(\Delta+1)\max_c p_c^{1-a_{\rm H}}
 \le\frac{\mathrm e2^{1/4}}8<1.
\]
Here $\Delta+1\le Q$. Each local conditional probability requires
at most $D^k$ predicate evaluations and polynomial input-preparation
work. Thus the finder costs
$(n+m)^{O(k\log D)}D^{O(k)}$, makes no counting calls, and fits the
claimed additional-work bound.


\begin{thebibliography}{99}
\small
\bibitem{JVV}
M.~R. Jerrum, L.~G. Valiant, and V.~V. Vazirani.
Random generation of combinatorial structures from a uniform distribution.
\emph{Theoretical Computer Science} 43 (1986), 169--188.

\bibitem{CountingLLL}
H.~Liu, C.~Wang, Y.~Yin, Y.~Zhang, and C.~Zhou.
A counting Lov\'asz Local Lemma. arXiv:2608.08616v1, 9 August 2026.

\bibitem{Alon}
N.~Alon. A parallel algorithmic version of the local lemma.
\emph{Random Structures \& Algorithms} 2(4) (1991), 367--378.

\bibitem{WangYin}
C.~Wang and Y.~Yin. A sampling Lov\'asz local lemma for large domain sizes.
In \emph{FOCS} (2024), 129--150. arXiv:2307.14872.

\bibitem{HSS}
B.~Haeupler, B.~Saha, and A.~Srinivasan.
New constructive aspects of the Lov\'asz Local Lemma.
\emph{Journal of the ACM} 58(6) (2011), Article 28.

\bibitem{Harris}
D.~G. Harris. Deterministic algorithms for the Lov\'asz Local Lemma:
simpler, more general, and more parallel.
\emph{Random Structures \& Algorithms} 63(3) (2023), 716--752;
conference version in \emph{SODA} (2022).
We use arXiv:1909.08065v6 for the numbered statements.

\bibitem{Moitra}
A.~Moitra.
Approximate counting, the Lov\'asz local lemma, and inference in graphical models.
\emph{Journal of the ACM} 66(2), Article 10 (2019); conference version in \emph{STOC} (2017).
\href{https://arxiv.org/abs/1610.04317}{arXiv:1610.04317}.

\bibitem{BGGGS}
I.~Bez\'akov\'a, A.~Galanis, L.~A.~Goldberg, H.~Guo, and D.~\v{S}tefankovi\v{c}.
Approximation via correlation decay when strong spatial mixing fails.
\emph{SIAM Journal on Computing} 48(2), 279--349 (2019).
\href{https://arxiv.org/abs/1510.09193}{arXiv:1510.09193}.

\bibitem{EL}
P.~Erd\H{o}s and L.~Lov\'asz.
Problems and results on 3-chromatic hypergraphs and some related questions.
In \emph{Infinite and Finite Sets}, Colloquia Mathematica Societatis J\'anos Bolyai 10 (1975).

\bibitem{Shearer}
J.~B.~Shearer.
On a problem of Spencer.
\emph{Combinatorica} 5(3), 241--245 (1985).

\bibitem{Beck}
J.~Beck.
An algorithmic approach to the Lov\'asz local lemma. I.
\emph{Random Structures \& Algorithms} 2(4), 343--365 (1991).

\bibitem{Moser}
R.~A.~Moser.
A constructive proof of the Lov\'asz local lemma.
In \emph{STOC} (2009).
\href{https://arxiv.org/abs/0810.4812}{arXiv:0810.4812}.

\bibitem{MT}
R.~A.~Moser and G.~Tardos.
A constructive proof of the general Lov\'asz local lemma.
\emph{Journal of the ACM} 57(2), Article 11 (2010).
\href{https://arxiv.org/abs/0903.0544}{arXiv:0903.0544}.

\bibitem{HarrisMT}
D.~G.~Harris.
New bounds for the Moser--Tardos distribution.
\emph{Random Structures \& Algorithms} 57(1), 97--131 (2020).
\href{https://arxiv.org/abs/1610.09653}{arXiv:1610.09653}.

\bibitem{GJL}
H.~Guo, M.~Jerrum, and J.~Liu.
Uniform sampling through the Lov\'asz local lemma.
\emph{Journal of the ACM} 66(3), Article 18 (2019); conference version in \emph{STOC} (2017).
\href{https://arxiv.org/abs/1611.01647}{arXiv:1611.01647}.

\bibitem{QWZ}
G.~Qiu, Y.~Wang, and C.~Zhang.
A perfect sampler for hypergraph independent sets.
In \emph{ICALP} (2022).
\href{https://arxiv.org/abs/2205.02050}{arXiv:2205.02050}.

\bibitem{GLLZ}
H.~Guo, C.~Liao, P.~Lu, and C.~Zhang.
Counting hypergraph colorings in the local lemma regime.
\emph{SIAM Journal on Computing} 48(4), 1397--1424 (2019).
\href{https://arxiv.org/abs/1711.03396}{arXiv:1711.03396}.

\bibitem{FGYZ}
W.~Feng, H.~Guo, Y.~Yin, and C.~Zhang.
Fast sampling and counting $k$-SAT solutions in the local lemma regime.
\emph{Journal of the ACM} 68(6), Article 40 (2021).
\href{https://arxiv.org/abs/1911.01319}{arXiv:1911.01319}.

\bibitem{FHY}
W.~Feng, K.~He, and Y.~Yin.
Sampling constraint satisfaction solutions in the local lemma regime.
In \emph{STOC}, 1565--1578 (2021).
\href{https://arxiv.org/abs/2011.03915}{arXiv:2011.03915}.

\bibitem{JPVGeneral}
V.~Jain, H.~T.~Pham, and T.-D.~Vuong.
Towards the sampling Lov\'asz local lemma.
In \emph{FOCS} (2021).
\href{https://arxiv.org/abs/2011.12196}{arXiv:2011.12196}.

\bibitem{JPVAtomic}
V.~Jain, H.~T.~Pham, and T.-D.~Vuong.
On the sampling Lov\'asz local lemma for atomic constraint satisfaction problems.
\href{https://arxiv.org/abs/2102.08342}{arXiv:2102.08342}, 2021.

\bibitem{HWYSampling}
K.~He, C.~Wang, and Y.~Yin.
Sampling Lov\'asz local lemma for general constraint satisfaction solutions in near-linear time.
In \emph{FOCS}, 147--158 (2022).
\href{https://arxiv.org/abs/2204.01520}{arXiv:2204.01520}.

\bibitem{HWYCounting}
K.~He, C.~Wang, and Y.~Yin.
Deterministic counting Lov\'asz local lemma beyond linear programming.
In \emph{SODA}, 3388--3425 (2023).
\href{https://arxiv.org/abs/2212.14847}{arXiv:2212.14847}.

\bibitem{BDK}
M.~Bordewich, M.~Dyer, and M.~Karpinski.
Path coupling using stopping times and counting independent sets and colorings in hypergraphs.
\emph{Random Structures \& Algorithms} 32(3), 375--399 (2008).
\href{https://arxiv.org/abs/math/0501081}{arXiv:math/0501081}.

\bibitem{HSZ}
J.~Hermon, A.~Sly, and Y.~Zhang.
Rapid mixing of hypergraph independent sets.
\emph{Random Structures \& Algorithms} 54(4), 730--767 (2019).
\href{https://arxiv.org/abs/1610.07999}{arXiv:1610.07999}.

\end{thebibliography}
\end{document}